%% file: main-single_column.tex
\documentclass[12pt, onecolumn]{IEEEtran}
\usepackage{cite} 
\usepackage{tikz}

\usepackage{amsmath,amsfonts,amssymb,amsthm}

\usepackage{subcaption} 
\usepackage{hyperref}
\usepackage{caption}

\usepackage{mathtools}
\usepackage{cuted}
\usepackage{stfloats}

\renewcommand{\P}{\mathbb{P}}
\newcommand{\E}{\mathbb{E}}

\newcommand{\lbs}{\lambda_{BS}}
\newcommand{\lu}{\lambda_U}

\newcommand{\dr}{\text{d}r}

\newcommand{\dt}{\text{d}t}
\newcommand{\dx}{\text{d}x}

\newcommand{\dy}{\text{d}y}
\newcommand{\ds}{\text{d}s}

\newcommand{\ruimte}{\IEEEeqnarraynumspace}

\newcommand{\wtot}{W_{tot}}
\newcommand{\wtotarea}{\frac{\overline{\wtot}}{\lbs}}

\newcommand{\rlos}{r_{LoS}}

\newtheorem{theorem}{Theorem}

\begin{document}
\title{Impact of Multi-connectivity on Channel Capacity and Outage Probability in Wireless Networks}
%
%
%

\author{Lotte~Weedage, 
        Clara~Stegehuis, and 
        Suzan~Bayhan 
\thanks{Authors are with the Faculty
of Electrical Engineering, Mathematics and Computer Science~(EEMCS), University of Twente, The Netherlands,
Corresponding author's e-mail: l.weedage@utwente.nl.}
}

\maketitle

\begin{abstract}
Multi-connectivity facilitates higher throughput, shorter delay, and lower outage probability for a user in a wireless network. Considering these promises, a rationale policy for a network operator would be to implement multi-connectivity for all of its users. 
In this paper, we investigate whether the promises of multi-connectivity also hold in such a setting where all users of a network are connected through multiple links. In particular, we consider a network where every user connects to its $k$ closest base stations. 
Using a framework of stochastic geometry and probability theory, we obtain analytic expressions for per-user throughput and outage probability of $k$-connectivity networks under several failure models. 
In contrast to the conclusions of previous research
, our analysis shows that per-user throughput decreases with increasing $k$. However, multi-connected networks are more resilient against failures than single connected networks as reflected with lower outage probability and lead to higher fairness among the users. Consequently, we conclude that rather than implementing multi-connectivity for all users, a network operator should consider it for its users who would benefit from additional links the most, e.g., cell edge users. 
\end{abstract}


%
\IEEEpeerreviewmaketitle

\section{Introduction}\label{sec:introduction}
\input{1.introduction}

\section{Related Work}\label{sec:related}
\input{related}

\section{System Model}\label{sec:system_model}
\input{2.system_model}

\section{Capacity Analysis of Multi-connectivity}\label{sec:capacity}
\input{3.capacity_outage_analysis}

\section{Impact of Link Failures on Multi-connectivity Capacity and Outage Probability}\label{sec:link_failures}
\input{4.link_failures}

\section{Discussion}\label{sec:discussion}
\input{5.discussion}

\section{Conclusion}\label{sec:conclusion}
\input{6.conclusion}

\section*{Acknowledgements}

This work has been supported by University of Twente under the project EERI: Energy-Efficient and Resilient Internet. 

We thank A.J.E.M. Janssen for help with the calculations in Appendix A.

\bibliographystyle{IEEEtran}

\appendices
\section{Proof of Theorem 1}\label{apx:proof_theorem1}
\input{7.appendix}

\end{document}

%% file: 1.introduction.tex


\textit{Multi-connectivity}~(MC) refers to a setting in which a user is simultaneously connected to a network through multiple connections, e.g., base stations~(BS). As the traffic can flow through multiple paths, a user can enjoy higher data rates, higher reliability, or lower link delays, which are essential for emerging applications such as augmented-reality, which requires a high amount of bandwidth, or mission-critical applications demanding ultra-reliable and low latency communication~(URLLC) \cite{Simsek20165G-EnabledInternet}. In addition, MC can increase the resilience of a network by exploiting spatial or/and frequency diversity. For instance, a user can be connected to a millimeter wave~(mmWave) network which provides an ample bandwidth while a second link at sub-6GHz bands can ensure a stable connection in case the mmWave link is disrupted due to an object blocking the line-of-sight~(LoS) mmWave connection. Regarding spatial diversity, a user can be connected to multiple BSs which are at different locations thereby improving communication quality, e.g., via decreasing the likelihood of correlated large-scale fading or increasing the robustness against network failures~\cite{Wolf2019HowMulti-Connectivity}. 

Depending on the primary goal, one of the following approaches can be used to realize MC: (i) load balancing, (ii) packet splitting, and (iii) packet duplication~\cite{Suer2020Multi-ConnectivityOverview}. In \textit{load balancing}, packets are distributed among  different connections, which will decrease the latency in comparison to \textit{single connectivity}~(SC) networks. 
Similar to load balancing, \textit{packet splitting} distributes traffic over multiple links, but at the packet level. That is, every packet is split into multiple parts and then every part is sent over a different link. 
Consequently, packet splitting will reduce the latency, but will not improve reliability.
When the primary goal is to increase the reliability as in mission-critical applications, \textit{packet duplication} based MC suggests transmitting each packet over every link so that the likelihood of failure is decreased at the expense of capacity.

While the promise of MC over SC is demonstrated in the prior works such as \cite{petrov2017dynamic}, a key question remains to find the optimal \textit{degree} of MC, i.e., the number of connections per user, so that the throughput and network reliability is high.
%
Since MC results in various complexities such as signalling overhead among multiple network nodes, traffic scheduling or combining at the receiver~\cite{Wolf2019HowMulti-Connectivity}, it is desirable to find the degree of MC that achieves a significant improvement over a one-lower connectivity.   
To find this optimal degree, one can use different measures, such as the \textit{spectral efficiency}, \textit{outage probability}, and \textit{channel capacity} of the network.
In \cite{Gapeyenko2019OnDeployments}, the authors find the optimal number of connections to vary between 2 and 4 in mmWave networks, based on the outage probability for the cell-edge users and the spectral efficiency of the network. 
%
In this paper,  we are interested in throughput of a randomly-selected user (not necessarily at the cell-edge) as it gives more insight on whether an application's rate requirements can be met. When focusing on the perspective of a user, we consider per-user throughput as our optimality metric. Moreover, we focus on multi-connectivity in the form of load balancing: users are served through every connection they have without redundancy.

Given that another motivation for MC is to increase reliability, we investigate the gain facilitated by MC in the existence of failures in the network. These failures can be due to the loss of line of sight~(LoS) links as experienced in mmWave 
networks~\cite{Gapeyenko2019OnDeployments} or failures caused by wireless channel impairments and overloading of BSs. 
We will explore the impact of MC on the expected channel capacity in the face of failures as well as the outage probability which quantifies the probability that a user is disconnected from the network. 
More precisely, we will address the following questions by providing an analytical framework based on stochastic geometry and probability theory:
\begin{itemize}
    \item How does multi-connectivity affect per-user throughput and outage probability of a cellular network when all links are reliable?
    \item How does multi-connectivity affect per-user throughput and outage probability in case of various link and BS failures~(e.g., random)?
\end{itemize}

Interestingly, our analysis shows that going from single to multi-connectivity decreases  per-user throughput if all users are multi-connected with equal number of connections. Thus, while MC may increase spectral efficiency~\cite{Gapeyenko2019OnDeployments}, the fact that BSs share their resources among more users that are possibly further away makes MC decrease per-user throughput. Hence, a key take-away of our paper is that simply increasing the number of connections of a user may not improve user throughput and MC should be leveraged for users who would benefit from it the most such as for cell edge users or users with higher rate or reliability requirements. Another observation is that higher degrees of MC increase fairness among users and that MC significantly boosts the network outages under failures.

In the following sections, we elaborate on the models for MC~(Section~\ref{sec:system_model}) and we analytically find the expectation of the channel capacity of a user under MC~(Section~\ref{sec:capacity}). Then, we investigate the outage probability under different failure models~(Section~\ref{sec:link_failures}), which we can use to find the expected channel capacity after failures. We present a performance analysis of the considered models in Section~\ref{sec:perf_analysis} to investigate the impact of MC for wireless networks. Finally, we conclude the paper in Section~\ref{sec:conclusion} with some future research directions.  

%% file: related.tex

In this section, we overview the three lines of related work: impact of MC on the network performance, impact of failures on MC, and scheduling schemes for MC.  For further reading of the state-of-the art on MC, we refer the reader to \cite{Suer2020Multi-ConnectivityOverview}.



\textit{Performance and reliability}: The most important aspect of MC in this paper is the performance and reliability of the network under MC. The former is typically measured in terms of network throughput whereas the reliability is usually measured in terms of outage probability. 
In a scenario without mobility, most research shows that while MC always increases outage probability, the per-user throughput may decrease \cite{ghatak2020elastic, sharma2020statistical}. This is sometimes referred to as the rate-reliability trade-off \cite{wolf2018rate}. In these papers, the outage probability is defined as the probability that the SNR of a user is below a certain threshold. However, these papers use a dynamic number of links per user, which makes analytical expressions for coverage probability and throughput intractable. Therefore, these papers only provide analytical results for 1- and 2-connectivity, and present an analysis of the impact of higher degrees of connectivity via simulations. In an urban scenario, with line-of-sight blockages of links, MC increases spectral efficiency and reliability \cite{Gapeyenko2019OnDeployments}. However, the authors only focus on a single user, and do not investigate the per-user throughput of MC, which we analyze in this work.

\textit{Impact of failures:} Little research has been done on the impact of \textit{link failures} on multi-connected networks, where links can fail due to other causes than only low S(I)NR. In \cite{Gapeyenko2019OnDeployments}, the authors show  how the loss of line-of-sight affects the network throughput. Similarly, it has been shown that in ultra-dense urban scenarios \cite{petrov2017dynamic, perdomo2020user, gerasimenko2019capacity} and indoor scenarios \cite{pirmagomedov2019performance}, MC increases the per-user throughput. Our work differs from these papers in the sense that we consider not only loss of line-of-sight, but a more general class of link failures, based on for example distance or number of links to a certain BS.

\textit{Scheduling schemes}: As discussed before, the main scheduling approaches to achieve  MC are load balancing, packet duplication, or packet splitting. 
In \cite{suer2020evaluation, suer2020impact, suer2020reliability}, the authors investigate the impact of these three methods and show that for different configurations of the network, different scheduling schemes work best. Specifically, for a network that is close to its capacity limit, load balancing provides the lowest latency and highest reliability among all three methods. 
%
In this paper, we will assume load balancing based MC. However, our analytical modelling approach can easily be extended to packet splitting and packet duplication based MC. 

Summarizing, our main contributions are threefold:
\begin{itemize}
\item We provide an analytical framework to analyze $k$-connected MC networks for arbitrary $k$. 
\item We show that contrary to the single-user perspective, MC decreases per-user throughput. 
\item The framework enables to analyze the outage probability for a more general class of failure models, including low SNR failureas and line-of sight failures. We quantify to what extent MC increases network reliability under several types of failure models.
\end{itemize}

%% file: 2.system_model.tex
We consider a 5G cellular network where base stations~(BS) and the users~(U) are distributed following  a homogeneous Poisson Point Process~(PPP)~\cite{ElSawy2017ModelingTutorial}. 
Let us denote the corresponding parameters $\lbs$ and $\lu$ for the BSs and the users, respectively. 
As BS association scheme, we assume that each user in this network connects to exactly $k$ BSs which provide the highest signal-to-noise-ratio~(SNR). We refer to this MC setting as \textit{k-connectivity}. For example, Figure~\ref{fig:2_connectivity_example} depicts a 2-connectivity scenario.

The network operator owns a spectrum of $\wtot$ Hz and allocates this spectrum to the BSs such that neighboring BSs result in negligible interference to each other. 
Each BS is allocated equal amount of spectrum denoted by $W_{BS}$ and transmits with a power level of $P_{tx}$ Watts. 
We will denote the number of users connected to BS$_j$ by $D_{BS_j}$ and refer to it as the degree of this BS. Similarly, we will denote the number of connections of a user by $D_U$. Note that when $D_U=0$, the user is under outage. We will denote the outage probability by $\P(D_U = 0)$.

We denote the path loss exponent with $\alpha$ and the noise power with $\sigma$. Assuming that BSs are assigned an orthogonal frequency, interference from the other BSs is negligible. This implies that under k-connectivity, each user connects to the $k$ closest BSs. 

We will denote the distance of a randomly-selected user to its $j^{\textrm{th}}$ closest BS~(BS$_j$) by $R_j$ and denote the SNR from BS$_j$ by $SNR_j$. Consequently, we denote the channel capacity of the link between a user and BS$_j$ by $C_j$. 
We assume a load-balancing MC network in which downlink traffic for each user is distributed over multiple links~\cite{Suer2020Multi-ConnectivityOverview}. 
As a result, we define the channel capacity $C^k_{sum}$ as the sum of all channel capacities that a user has from its links with BS$_j, j \in [k]$ in $k$-connectivity. 

Let $N_{BS}$ denote the number of BSs within a circle with radius $r$. Similarly, let $N_{U}$ denote the number of users within a circular area with radius $r$. 
With the PPPs for the spatial distributions of BSs and users, we can calculate the probability of $n$ BSs/users in the considered circular area as follows:
\begin{IEEEeqnarray}{rCl}
    P(N_{BS}(\pi r^2) = n) &=& \frac{(\lbs \pi r^2)^n}{n!}e^{-\lbs \pi r^2}, \label{eq:number_of_basestations}\\
    P(N_{U}(\pi r^2) = n) &=& \frac{(\lu \pi r^2)^n}{n!}e^{-\lu \pi r^2}.\label{eq:number_of_users} 
\end{IEEEeqnarray}

\begin{figure}
    \centering
    \resizebox{0.4\textwidth}{!}{%
        \begin{tikzpicture}

        \node[inner sep = -2pt] (BS1) at (1,0.5) {\includegraphics[scale = 0.05]{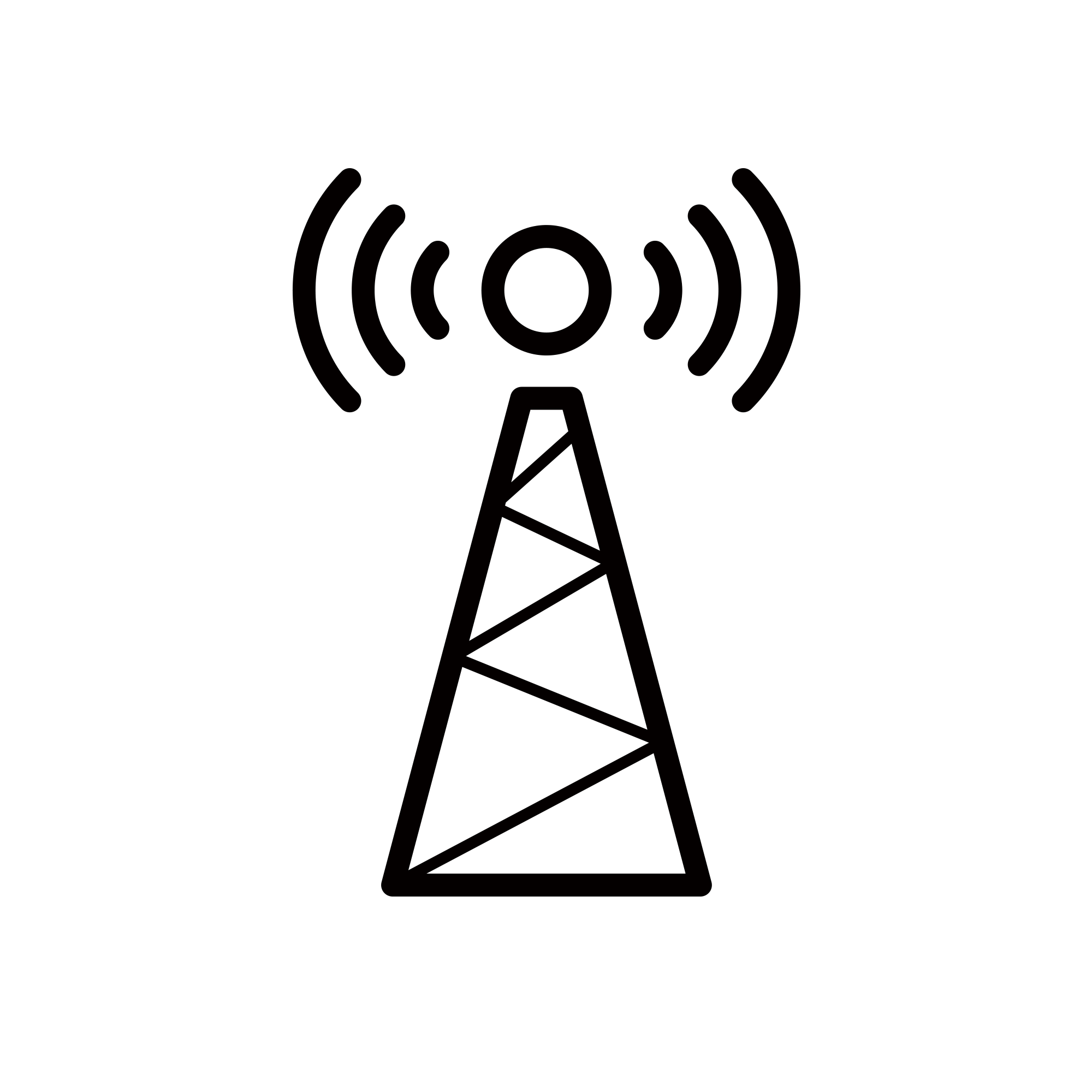}};
        \node[inner sep = -2pt] (BS2) at (-2.5,3) {\includegraphics[scale = 0.05]{base_station_icon.png}};
        \node[inner sep = -2pt] (BS3) at (3,-1.5) {\includegraphics[scale = 0.05]{base_station_icon.png}};
        \node[inner sep = -2pt] (BS4) at (-2,-2) {\includegraphics[scale = 0.05]{base_station_icon.png}};

        \node[shape = circle, fill, color = green, scale =0.5] (U1) at (-2,1.5) {};
        \node[shape = circle, fill, color = green, scale =0.5]  (U2) at (-1,2) {};
        \node[shape = circle, fill, color = green, scale =0.5] (U3) at (2,2.5) {};
        \node[shape = circle, fill, color = green, scale =0.5]  (U4) at (0,3.5) {};
        \node[shape = circle, fill, color = green, scale =0.5] (U5) at (-2.5,0.5) {};
        \node[shape = circle, fill, color = green, scale =0.5] (U6) at (-1.5,-1) {};
        \node[shape = circle, fill, color = green, scale =0.5] (U7) at (2,-3) {};
        \node[shape = circle, fill, color = green, scale =0.5]  (U8) at (-1.5,-3) {};
        \node[shape = circle, fill, color = green, scale =0.5] (U9) at (2,-1) {};
        \node[shape = circle, fill, color = green, scale =0.5] (U10) at (4,-2) {};
        \node[shape = circle, fill, color = green, scale =0.5] (U11) at (2,0) {};
        \node[shape = circle, fill, color = green, scale =0.5] (U12) at (3.5,0.5) {};
        \node[shape = circle, fill, color = green, scale =0.5] (U13) at (-3.5,-1.5) {};
        \node[shape = circle, fill, color = green, scale =0.5] (U14) at (-4,-0.5) {};
        \node[shape = circle, fill, color = green, scale =0.5] (U15) at (4.5,1.5) {};
        \node[shape = circle, fill, color = green, scale =0.5] (U16) at (-3.5,3) {};
        \node[shape = circle, fill, color = green, scale =0.5] (U17) at (-1,0.5) {};
        \node[shape = circle, fill, color = green, scale =0.5] (U18) at (-0.5,0) {};
        \node[shape = circle, fill, color = green, scale =0.5] (U19) at (1,-2) {};
        
        \draw[color = lightgray] (U16) -- (BS2);
        \draw[color = lightgray] (U1) -- (BS2);
        \draw[color = lightgray] (U2) -- (BS2);
        \draw[color = green, thick] (U4) -- (BS2);
        \draw[color = lightgray] (U5) -- (BS2);
        
        \draw[color = lightgray] (U14) -- (BS4);
        \draw[color = lightgray] (U13) -- (BS4);
        \draw[color = lightgray] (U8) -- (BS4);
        \draw[color = lightgray] (U6) -- (BS4);
        
        \draw[color = lightgray] (U18) -- (BS1);
        \draw[color = lightgray] (U17) -- (BS1);
        \draw[color = lightgray] (U11) -- (BS1);
        \draw[color = lightgray] (U3) -- (BS1);
        
        \draw[color = lightgray] (U15) -- (BS3);
        \draw[color = lightgray] (U12) -- (BS3);
        \draw[color = lightgray] (U9) -- (BS3);
        \draw[color = lightgray] (U7) -- (BS3);
        \draw[color = lightgray] (U10) -- (BS3);
        \draw[color = lightgray] (U19) -- (BS3);
        
        \draw [color = lightgray](U13) -- (BS2);
        \draw [color = lightgray](U14) -- (BS2);
        \draw [color = lightgray](U17) -- (BS2);
        
        \draw [color = lightgray](U5) -- (BS4);
        \draw [color = lightgray](U18) -- (BS4);
        \draw [color = lightgray](U19) -- (BS1);
        
        \draw [color = lightgray](U1) -- (BS1);
        \draw [color = lightgray](U16) -- (BS1);
        \draw [color = green, thick](U4) -- (BS1);
        \draw [color = lightgray](U2) -- (BS1);
        \draw [color = lightgray](U9) -- (BS1);
        \draw [color = lightgray](U6) -- (BS1);
        \draw [color = lightgray](U8) -- (BS1);
        \draw [color = lightgray](U12) -- (BS1);
        \draw [color = lightgray](U15) -- (BS1);
        \draw [color = lightgray](U10) -- (BS1);
        \draw [color = lightgray](U7) -- (BS1);
        
        \draw [color = lightgray](U3) -- (BS3);
        \draw [color = lightgray](U11) -- (BS3);
        \node at (-1.5,3.5) {$R_1$};
        \node at (0.6,2.6) {$R_2$};
        
        \end{tikzpicture}%
        }
        \caption{Users connect to the two nearest base stations.}
    \label{fig:2_connectivity_example}
\end{figure}

%% file: 3.capacity_outage_analysis.tex
In this section, we derive the expected channel capacity of a user denoted by ~$C^k_{sum}$. 

\subsection{Channel capacity under k-connectivity}
The Shannon channel capacity $C_j$ \cite{shannon2001mathematical} of a user that is connected to BS$_j$ depends on the bandwidth $W$ that BS$_j$ can allocate to this user and the SNR$_j$ at the considered user. More formally, $ C_{j}$ is defined as:
\begin{IEEEeqnarray}{rCl}
    C_{j} &=& {W}\log_2\left(1+SNR_{j}\right).\label{eq:Cj}
\end{IEEEeqnarray}

Given the degree of BS$_j$ and assuming that all users are allocated an identical amount of bandwidth, we can calculate $W$ simply as follows: 
\begin{align}
    W = \frac{W_{BS}}{D_{BS_j}}.\label{eq:W_j}
\end{align} 

To find the bandwidth $W_{BS}$ allocated to each BS, we need to know the number of BSs deployed in the coverage area~($A$) of the operator. Then, we derive $W_{BS}$ as:
\begin{IEEEeqnarray}{rCl}
    W_{BS} &= \frac{\wtot}{\lbs \cdot A} = \frac{\overline{W_{tot}}}{\lbs}, \label{eq:W_BS_j}
\end{IEEEeqnarray}
where $\overline{\wtot}$ is the total bandwidth per area.

We can calculate $SNR_j$ as follows:
\begin{IEEEeqnarray}{rCl}
    SNR_{j} = \frac{P_{tx} R_{j}^{-\alpha}}{\sigma} = \begin{cases} c R_{j}^{-\alpha} \hspace{1cm} R_{j} \geq 1, \\ c \hfill R_{j} < 1,\end{cases} \label{eq:definition_snr}
\end{IEEEeqnarray}
for $c = \frac{P_{tx}}{\sigma}$. Inserting (\ref{eq:W_j}) and (\ref{eq:W_BS_j}) into (\ref{eq:Cj}), we obtain $C_j$ as:
\begin{IEEEeqnarray}{rCl}
    C_{j} &= \frac{\overline{\wtot}}{\lbs D_{BS_j}}\log_2\left(1+SNR_{j}\right) \label{eq:channel_capacity}.
\end{IEEEeqnarray}

Consequently, we define $C^k_{sum}$ as the sum of all channel capacities provided by all $k$ links of this user:
\begin{IEEEeqnarray}{rCl}
    C^k_{sum} &=& \sum_{j=1}^k \frac{\overline{\wtot}}{\lbs D_{BS_j}}  \log_2(1+SNR_{j}) = \frac{\overline{\wtot}}{\lbs D_{BS}} \sum_{j=1}^k  \log_2(1+SNR_{j}) \label{eq:sum_channel_capacity},
\end{IEEEeqnarray}
as we assume $D_{BS}$ is independent of $j$ under the assumption that users are  distributed following a homogeneous PPP.


To calculate $C^k_{sum}$ in (\ref{eq:sum_channel_capacity}), we will derive the expected degree of a BS~($D_{BS}$), the distribution of SNR, and finally the expectation of $\log_2(1+SNR)$ in the following sections.

\subsection{Expected degree of base station}
We can approximate the degree distribution of a BS by finding the area in which users connect to this BS in $k-$connectivity, which has been derived in \cite{stegehuis2021degree}:
\begin{IEEEeqnarray}{rCl}
    \P(D_{BS} = n) &=& \frac{\Gamma(n + a_k)}{\Gamma(n+1)\Gamma(a_k)} \frac{a_k^{a_k} (k \lambda)^n}{(k\lambda + a_k)^{a_k {+} n}} ,\IEEEeqnarraynumspace\label{eq:degree_distribution}
\end{IEEEeqnarray}
for $a_1 = 3.5, a_2 = 7.2, a_3 = 11.1, a_4 = 15.2$ and $a_5 = 21.2$ and where $\lambda = \frac{\lu}{\lbs}$.

To obtain the expectation of 
$\frac{\overline{\wtot}}{\lbs D_{BS}}$ in (\ref{eq:sum_channel_capacity}), we use the size-biased distribution of $D_{BS}$, which is the distribution of the degree of a BS to which a randomly chosen user connects:
\begin{IEEEeqnarray}{rCl}
    \E\left(\frac{\overline{\wtot}}{\lbs D^*_{BS}}\right) &=& \frac{\overline{\wtot}}{\lbs} \sum_{n=1}^\infty \frac{1}{n}\P(D^*_{BS} = n)\nonumber \\
    &=& \wtotarea \sum_{n=1}^\infty \frac{\P(D_{BS} = n)}{\E(D_{BS})} \nonumber \\
    &=& \wtotarea \frac{1}{k\lambda} = \frac{\overline{\wtot}}{k \lu}. \label{eq:expectation_w_general}
\end{IEEEeqnarray}

\subsection{Probability distribution of SNR}
Since the SNR depends on the distance between a user and the BS in consideration, we need to find the distance distribution from a randomly chosen user to the $j^{\textrm{th}}$ closest BS for all $j \in [k]$. 
 We derive the probability distribution of $R_j$, which is defined as the distance to the $j^{\textrm{th}}$ closest BS~(Fig.~\ref{fig:2_connectivity_example}). 
 To derive this distribution, we use the fact that the $j^{\textrm{th}}$ closest BS is at a distance larger than $r$ if and only if there are less than $j$ BSs within the surrounding area $\pi r^2$ of the considered user. Formally, we state this probability as follows:
\begin{IEEEeqnarray}{rCl}
    \P(R_j > r) & =& \P(N_{BS}(\pi r^2) \leq j-1) \nonumber \\
    &=& \sum_{i = 0}^{j-1} \frac{(\lbs \pi r^2)^{i}}{\Gamma(i+1)}e^{-\lbs \pi r^2} \nonumber \\
    &=& \frac{\Gamma(j, \lbs \pi r^2)}{\Gamma(j)}, \label{eq:distance_probability}
\end{IEEEeqnarray}
where we used \eqref{eq:number_of_basestations} and $\Gamma(z, x) = \int_x^\infty t^{z-1}e^{-t}\dt$ is the upper incomplete gamma function. Therefore, we can write the distribution of $R_j$ as:
\begin{IEEEeqnarray}{rCl}
    f_{R_j}(r) &= \frac{2\left(\lbs \pi r^2\right)^j}{r\Gamma(j + 1)}e^{-\lbs\pi r^2}. \label{eq:distance_pdf}\nonumber
\end{IEEEeqnarray}
 With this distance distribution, we can find the probability distribution of the SNR of the $j^{\textrm{th}}$ closest BS:
\begin{IEEEeqnarray}{rCl}
    \P\left(SNR_j \leq x \right) &=  \P\left(R_j \geq \left(\frac{c}{x}\right)^{\frac{1}{\alpha}}\right) =  \frac{\Gamma\left(j, \phi x^{-\frac{2}{\alpha}}\right)}{\Gamma(j)},\IEEEeqnarraynumspace\nonumber
\end{IEEEeqnarray}
where $\phi = \lbs \pi c^{\frac{2}{\alpha}}$ and therefore, 
\begin{IEEEeqnarray}{rCl}
    f_{SNR_j}(x) \!&=\! \frac{d}{dx}\!\left[\frac{\Gamma(j, \phi x^{-\frac{2}{\alpha}})}{\Gamma(j + 1)}\right]\! =\! \frac{2\left(\phi x^{-\frac{2}{\alpha}}\right)^j}{\alpha x\Gamma(j)}e^{-\phi x^{-\frac{2}{\alpha}}}.\IEEEeqnarraynumspace\nonumber
\end{IEEEeqnarray}
Since the SNR cannot become larger than $c$, the probability distribution of the SNR becomes:
\begin{IEEEeqnarray}{rCl}
f_{SNR_j}(x) =
    \begin{cases}
    \frac{2\left(\phi x^{-\frac{2}{\alpha}}\right)^j}{\alpha x\Gamma(j)}e^{-\phi x^{-\frac{2}{\alpha}}}, \hspace{2cm} x < c\\
    \P(R_j < 1) = 1 - \frac{\Gamma(j, \lbs\pi)}{\Gamma(j)}, \hfill x = c\\
    0 \hfill x > c.
    \end{cases}\label{eq:snr_pdf}
\end{IEEEeqnarray}

\begin{figure}[th]
\centering
    \includegraphics[width=0.4\textwidth]{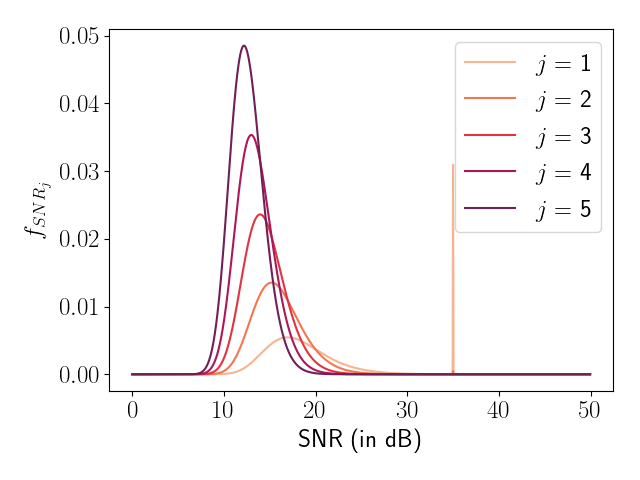}  
    \caption{Probability density function of $SNR_j$ for different values of $j$, with $\lbs = 10^{-2}$, $c = 10^{3.5}$ and $\alpha = 2$.}
    \label{fig:SNR_pdf}
\end{figure}

Fig.~\ref{fig:SNR_pdf} depicts $f_{SNR_j}$ under different values of $j$. As the figure shows, the SNR decreases when BSs are further away (larger $j$), and the SNR is cut off at $c = 10^{3.5} = 35$ dB, as users closer than $1$m from the BS will have power/noise ratio according to (\ref{eq:definition_snr}).

\subsection{Expectation logarithm of $SNR_j$}
To find the expected channel capacity, we need to find the expectation of $\log_2(1+SNR_j)$ for $j \in [k]$. We define this expectation in Theorem \ref{thm:expectation_log_snr} and provide an approximation of it that holds in high SNR regime, i.e., when the BS transmission power $P_{tx}$ is large and consequently $c$ is large.
\begin{theorem}\label{thm:expectation_log_snr}
The expectation of $\log_2(1 + SNR_j)$ 
is defined as:
\begin{IEEEeqnarray}{rCl}
    \E(\log_2(1+SNR_j)) &=& \frac{1}{\ln(2)\Gamma(j)} G(j, \phi) +\log_2(1+c) \left(1 - \frac{\Gamma(j, \lbs \pi)}{\Gamma(j)}\right) \nonumber \\
    &&+ \frac{\alpha}{2 \ln(2) \Gamma(j)}\bigg(\ln(\phi)\left(\Gamma(j, \lbs \pi) - \Gamma(j, \phi)\right) - \frac{\text{d}}{\text{d}a} \left[ \Gamma(a, \lbs \pi) - \Gamma(a, \phi)\right]_{a = j} \bigg),\label{eq:SNR_theorem}\IEEEeqnarraynumspace
\end{IEEEeqnarray}
where $\phi = \lbs \pi c^{\frac{2}{\alpha}}$, and $G(j, \phi)$ is defined as follows:
    \begin{IEEEeqnarray}{rCll}
    G(j, \phi) &=& \sum_{i=0}^\infty \frac{(-1)^i\phi^{\frac{\alpha}{2}(i+1)}}{i + 1} \Gamma\left(\frac{-\alpha}{2}(i{+}1){+}j, \phi\right)\nonumber \\
    &&{+} \frac{(-1)^i\phi^{-\frac{\alpha}{2}(i+1)}}{i + 1}  \bigg(\Gamma\left(\frac{\alpha}{2}(1{+}i) {+} j, \lbs \pi \right) {-}  \Gamma\left(\frac{\alpha}{2}(1{+}i) {+} j, \phi\right)\bigg).\IEEEeqnarraynumspace\label{eq:GJ}
\end{IEEEeqnarray}

Moreover, for large $c$, $\lbs < \frac{1}{\pi}$ and $\phi > 1$, which corresponds to the high SNR regime, this expectation can be approximated with:
\begin{IEEEeqnarray}{rCl}
    \E(\log_2(1+SNR_1)) &=& \frac{\alpha}{2 \ln(2)}\left(\ln(\phi) + \gamma - \lbs \pi \right)  + \delta_1,\IEEEeqnarraynumspace\label{eq:approx_snr_k1_thm}\\
    \E(\log_2(1+SNR_j)) &=& \frac{\alpha }{2 \ln(2)} \left(\ln(\phi) + \gamma - H_{j-1}\right) + \delta_j, \IEEEeqnarraynumspace\label{eq:approx_snr_thm}
\end{IEEEeqnarray}
for $j > 1$ and where $\gamma$ is Euler's constant, $H_j = \sum_{i=1}^j \frac{1}{i}$ is the harmonic number and for the following variables $\delta_1$ and $\delta_j$:
\begin{IEEEeqnarray}{rCl}
    \delta_1 &=&  O\big(\phi^{-1}\big) {+} O\big((\lbs \pi)^2\big) {+} O\left(\ln(c) c^{-1}\right) {+} O\left(\E(SNR_1^{-1})\right),\nonumber \IEEEeqnarraynumspace\\
    \delta_j &=& O\left(c^{-\frac{2j}{\alpha}} \ln(\phi)\right) + O\left(\Gamma(j)^{-1} c^{-\frac{2j}{\alpha}}\phi^{-j}\right) + O\left(\ln(c)c^{-1}\right) + O\left(\E(SNR_j^{-1})\right), \hspace{1cm} j > 1. \nonumber
\end{IEEEeqnarray}
\end{theorem}

As can be seen in Fig.\ref{fig:I3_approximation}, the approximation given in Theorem~\ref{thm:expectation_log_snr} matches the real capacity in high SNR regime. Moreover, the plot shows that the relative differences of $\E(\log_2(1 + SNR_j))$ for $j = 1$ and $j = 5$ become lower under dense deployment of BSs represented by higher $\lbs$.
\begin{figure}[th]
    \centering
    \includegraphics[width = 0.4\textwidth]{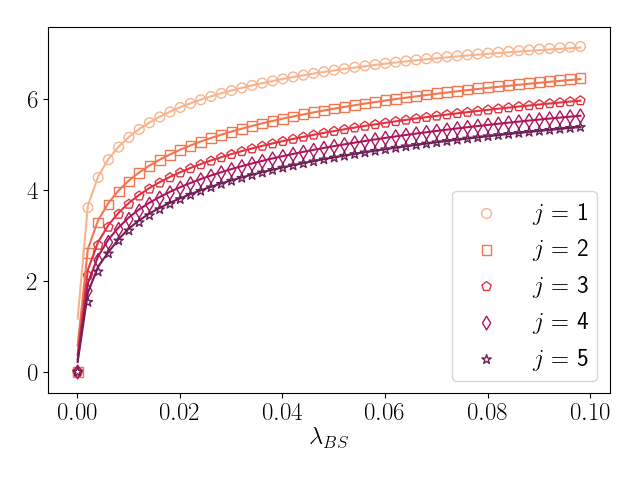}
    \caption{Real value (line) and approximation (circles) of $\E(\log_2(1+SNR_j))$ with $c = 12800$ and $\alpha = 2$.}
    \label{fig:I3_approximation}
\end{figure}

\subsection{Expected channel capacity}
By Theorem \ref{thm:expectation_log_snr} and \eqref{eq:expectation_w_general}, we now have an expression for the expected channel capacity of a single user that connects to BS$_j$ with a total of $k$ connections per user:
\begin{IEEEeqnarray}{rCl}
    \E(C_j) &= \frac{\overline{\wtot}}{k\lu} \cdot \E\left(\log_2(1+SNR_j)\right), \label{eq:expectation_channel_j}
\end{IEEEeqnarray}
and its approximation, using \eqref{eq:approx_snr_k1_thm} and \eqref{eq:approx_snr_thm}:
\begin{IEEEeqnarray}{rCll}
    \E(C_1) &=& \frac{\overline{\wtot}}{k\lu} \cdot\frac{\alpha}{2 \ln(2)}\left(\ln(\phi) +\gamma - \lbs \pi \right) + \delta_1, &\\
    \E(C_j) &=& \frac{\overline{\wtot}}{k\lu} \cdot \frac{\alpha }{2 \ln(2)} \left(\ln(\phi) + \gamma - H_{j-1}\right) + \delta_j, &j > 1.\IEEEeqnarraynumspace\label{eq:expected_channel_approx_jgeq1}
\end{IEEEeqnarray}

Note that we assume that the expected degree of a BS and the distance to this BS are independent. While this assumption is realistic for larger values of $k$ where the degree distribution becomes more concentrated~\cite{stegehuis2021degree}, it may not always hold as we will show in Section~\ref{sec:link_failures}. 

For $k > 1$, we are interested in the sum of the channel capacities, $C^k_{sum}$, as defined in \eqref{eq:sum_channel_capacity}:
\begin{IEEEeqnarray}{rCl}
    \E(C^k_{sum}) &=& \sum_{j=1}^k \E(C_j) \nonumber \\
    &=& \frac{\overline{\wtot}}{k\lu} \frac{\alpha}{2 \ln(2)} \Big(k\ln(\phi) + k\gamma -\lbs \pi - \sum_{j=1}^{k-1} H_{j}\Big)  \nonumber \\ 
    &=& \frac{\overline{\wtot}}{k\lu}\frac{\alpha}{2 \ln(2)} \left(k\ln(\phi) + k\gamma -\lbs \pi - k H_{k} + k\right) \nonumber \\
    &=& \frac{\overline{\wtot}}{\lu} \frac{\alpha}{2 \ln(2)} \Big( \ln(\phi) + \gamma  + 1-  \Big(H_k + \frac{\lbs \pi}{k}\Big)\Big),\label{eq:c_sum_approx}
\end{IEEEeqnarray}
where we used the property of harmonic numbers: $\sum_{j=1}^{k-1}H_j = kH_k - k$.

\begin{theorem}\label{th:decreasing_snr}
    In the high SNR regime, the sum of the channel capacities, $\E(C^k_{sum})$, as defined in \eqref{eq:c_sum_approx}, is decreasing in $k$.
\end{theorem}

\begin{proof}
    As a large part of \eqref{eq:c_sum_approx} does not depend on $k$, we only need to focus on the part $H_k + \frac{\lbs \pi}{k}$. 
    %
    Assume $k \geq 1, k \in \mathrm{N}$. Then, for $k + 1$:
    \begin{align}
        H_{k + 1} + \frac{\lbs \pi}{k+1}  = H_{k} + \frac{1 + \lbs \pi}{k+1} > H_{k} + \frac{\lbs \pi}{k},\nonumber
    \end{align}
    which holds for $\lbs \pi < 1$. Thus, for $k + 1$, we subtract by a larger value in~\eqref{eq:c_sum_approx}, and thus $\E(C^k_{sum})$ decreases in $k$.
\end{proof}

We use $\E(C_{sum}^k)$ as a measure to compare different degrees of MC under different types of failures in the following section.

%% file: 4.link_failures.tex
We now investigate the performance of MC networks under the following failure models, some of which are illustrated in Figure~\ref{tikz:failure_models}:
\begin{enumerate}
    \item \textbf{Random failures}: Every BS-U link fails with a probability $p$.
    \item \textbf{Overload failures}: BSs fail with a probability proportional to their degree. This scenario might reflect denial-of-service attacks in which a BS becomes non-functional due to many requests.
    \item \textbf{Distance failures}: Links fail when the distance between the user and the BS is higher than a certain threshold.
    \item \textbf{Line-of-sight failures}: Links fail because they are not in line-of-sight of a user.
\end{enumerate}

 We consider the sum of the channel capacities of a user who is connected to BS $1, 2, 3, \ldots, k$, where BS$_1$ is the closest and BS$_k$ is the furthest one. The expected channel capacity under failures can be written as:
\begin{IEEEeqnarray}{rCl}
    \E(C^k_{sum}) &= \sum_{j=1}^k \E\left((1- p_j)  C_j\right),\label{eq:channel_capacity_with_failure}
\end{IEEEeqnarray}
where $p_j$ is the probability that the link to BS$_j$ fails and $C_j$ is defined in \eqref{eq:channel_capacity}. Below, we derive analytic expressions for $p_j$ for the above-listed failure models. The second quantity we investigate is the outage probability, $\P(D_U = 0)$, defined as the probability that a user has no connections at all. 

Moreover, we consider two cases: no re-allocation or instantaneous re-allocation of bandwidth after a failure happens. In the case of instantaneous re-allocation, we assume that the BS re-allocates the bandwidth previously allocated to the failed link(s) among its currently active links. This means that in the expected channel capacity \eqref{eq:sum_channel_capacity}, we divide by the degree of the BS after failures, which leads to an increase in per-user bandwidth. Meanwhile, in the case without re-allocation, per-user bandwidth remains the same as the bandwidth is not re-allocated. 
In the following, we find exact expressions for the expected channel capacity and the outage probability without re-allocation. After that, we provide numerical results on failures with and without re-allocation.

\input{tikz-failure-models}

\subsection{Random failures}
Under random failures, every link fails with probability $p$. Therefore,~\eqref{eq:channel_capacity_with_failure} becomes:
\begin{IEEEeqnarray}{rCl}
    \E(C^k_{sum}) &=&  \frac{\overline{\wtot}}{k\lu} (1-p) \sum_{j=1}^k \E\left(\log_2(1+SNR_j)\right), \label{eq:random_channel}
\end{IEEEeqnarray}
where we obtain $\E(\log_2(1+SNR_j))$ from Theorem~\ref{thm:expectation_log_snr}.

The outage probability under random failures is:
\begin{IEEEeqnarray}{rCl}
    \P(D_U = 0) = \P(\text{every link fails}) = p^k.\label{eq:random_outage}
\end{IEEEeqnarray}

\subsection{Overload failures}
In the case of an overload failure, BSs fail with a probability proportional to the number of users who are connected to that BS. We assume that failures happen with probability $p$, defined as:
\begin{IEEEeqnarray}{rCl}
    p = \begin{cases}
    0, \hspace{2cm}& D_{BS} = 0,\\
    1 - (D_{BS})^{-\beta}, &D_{BS}>0,
    \end{cases}
\end{IEEEeqnarray}
for $\beta > 0$, where $D_{BS}$ is the degree of a BS. 
Thus, the probability that a given link from a user fails, is given by $1 - (D_{BS}^*)^{-\beta}$ where $D^*_{BS}$ is the size-biased distribution of $D_{BS}$. 
Thus, the expected channel capacity is:
\begin{IEEEeqnarray}{rCl}
    \E(C^k_{sum}) &=& \sum_{j=1}^k \E\left( \left(\left(D^*_{BS}\right)^{-\beta}\right) C_j\right) \nonumber\\
    &=& \sum_{j=1}^k \E\left( \left(D^*_{BS}\right)^{-\beta} \frac{\overline{\wtot}}{\lbs D^*_{BS}}  \log_2(1+SNR_j)\right)\nonumber \\
    &=& \wtotarea  \E\left(\left(D^*_{BS}\right)^{-\beta-1}\right)\sum_{j=1}^k \E(\log_2(1+SNR_j).\ruimte
    \label{eq:overload_failure}
\end{IEEEeqnarray}

We obtain the expectation of $\left(D_{BS}^*\right)^{-\beta - 1}$ similar to \eqref{eq:expectation_w_general}:
\begin{IEEEeqnarray}{rCl}
    \E\left(\left(D^*_{BS}\right)^{-\beta-1}\right) &=& \sum_{n= 1}^\infty n^{-\beta - 1} \P(D_{BS}^* = n) \nonumber \\
    &=& \frac{1}{k \lambda} \sum_{n=1}^\infty n^{-\beta} \P(D_{BS} = n) \nonumber \\
    &=& \frac{1}{k \lambda}\sum_{n=1}^\infty \frac{\Gamma(n+a_k)}{\Gamma(n+1)\Gamma(a_k)} \frac{a_k^{a_k}(k\lambda)^{n}}{(k\lambda + a_k)^{n+a_k}}n^{-\beta},
\end{IEEEeqnarray}
for $\lambda = \lu/\lbs$ and $\P(D_{BS} = n)$ as defined in \eqref{eq:degree_distribution}. Therefore,~\eqref{eq:overload_failure} equals:
\begin{IEEEeqnarray}{rCl}
    \E(C^k_{sum}) &=& \frac{\overline{\wtot}}{k \lu} \Bigg(\sum_{n=1}^\infty \frac{\Gamma(n+a_k)}{\Gamma(n+1)\Gamma(a_k)} \frac{a_k^{a_k}(k\lambda)^{n}}{(k\lambda + a_k)^{n+a_k}}n^{-\beta}\Bigg)  \sum_{j=1}^k \E\left(\log_2(1+SNR_j)\right),\label{eq:overload_channel} \IEEEeqnarraynumspace
\end{IEEEeqnarray}
where we obtain $\E(\log_2(1+SNR_j))$ from Theorem~\ref{thm:expectation_log_snr}.

Subsequently, the outage probability is:
\begin{IEEEeqnarray}{rCl}
    \P(D_U = 0)&=& \prod_{j=1}^k  \E\left(1 - \left(D_{BS}^*\right)^{-\beta}\right)=
     \Bigg(1 -\frac{1}{k\lambda} \sum_{n=1}^\infty \frac{\Gamma(n+a_k)}{\Gamma(n+1)\Gamma(a_k)} \frac{a_k^{a_k}(k\lambda)^{n}}{(k\lambda + a_k)^{n+a_k}}n^{-\beta}\Bigg)^k, \label{eq:overload_outage}\ruimte
\end{IEEEeqnarray}
as BSs fail independently from each other.

\subsection{Distance failures}
In this case, a link fails when the distance between the user and its associated BS is larger than $r_{max}$. This distance $r_{max}$ can be interpreted as the maximum distance such that the SNR is above a certain threshold. Then, the probability that link $j$ fails, is as follows:
\begin{IEEEeqnarray}{rCl}
    p_j &=& \P(R_j > r_{max}) = \frac{\Gamma(j, \lbs \pi r_{max}^2)}{\Gamma(j)},
\end{IEEEeqnarray}
where we used \eqref{eq:distance_probability}. This gives the following expression for $\E\left(C^k_{sum}\right)$:

\begin{IEEEeqnarray}{rCl}
     \E(C^k_{sum}) &=&  \frac{\overline{\wtot}}{k \lu} \sum_{j=1}^k (1{-}p_j) \E\left(\log_2(1{+}SNR_j) | R_j \leq r_{max}\right) \nonumber \\
     &=& \frac{\overline{\wtot}}{k \lu}\sum_{j=1}^k \P(R_j \leq r_{max})  \frac{\E(\log(1+SNR_j), R_j \leq r_{max})}{\P(R_j \leq r_{max})}\nonumber \\
     &=& \frac{\overline{\wtot}}{k \lu}\sum_{j=1}^k\Bigg( \int_1^{r_{max}} \log_2(1+c r^{-\alpha}) f_{R_j}(r)\dr + \log_2(1+c)\left(1 - \frac{\Gamma(j, \lbs \pi)}{\Gamma(j)}\right)\Bigg),\ruimte\label{eq:distance_channel}
\end{IEEEeqnarray}
where the last term of this expression comes from the case $R_j \leq 1$ in the definition of the SNR in \eqref{eq:definition_snr}, in which case the SNR is equal to $c$.

For the outage probability, we calculate the probability that the first link (hence from the closest BS) of a user fails.  Since all other links are established with  more distant BSs, they will also fail when the first link fails. Therefore, we calculate the outage probability as follows:
\begin{IEEEeqnarray}{rCl}
    \P(D_U = 0) 
    = \P(R_1 \geq r_{max}) = e^{-\lbs\pi r_{max}^2}.\label{eq:distance_outage}\IEEEeqnarraynumspace
\end{IEEEeqnarray}

\subsection{LoS failures}
We now investigate LoS failures, where a link between a user and a BS fails if the link is non-LoS.
In \cite{Gapeyenko2019OnDeployments}, the authors derive the probability that a connection is non-LoS for a mmWave network as follows:
\begin{IEEEeqnarray}{rCl}
    p_j &=& \begin{cases}
    0, \hfill R_j \leq \rlos,\\
    1- R_j^{-1}\left(\rlos+R_j e^{\frac{-R_{j}}{2 \rlos}}-18e^{\frac{-R_j}{2 \rlos}}\right), R_j > \rlos,
    \end{cases}\nonumber
\end{IEEEeqnarray}
where $\rlos$ is a parameter that describes at what distance objects are more likely to become non-LoS.
This LoS failure only depends on the distance between the user and the BS. We define:
\begin{equation}
    g(r)=1- r^{-1}\left(\rlos+r e^{\frac{-r}{2 \rlos}}-18e^{\frac{-r}{2\rlos}}\right).
\end{equation}
Then, the expected channel capacity under LoS failures equals:
\begin{IEEEeqnarray}{rCl}
    \E(C^k_{sum}) &=& \frac{\overline{\wtot}}{k \lu} \sum_{j=1}^k \E\left( (1-p_j) \log_2(1+SNR_j) \right)\nonumber\\
    &=& \frac{\overline{\wtot}}{k \lu} \sum_{j=1}^k \E\left( \log_2(1+SNR_j^{-\alpha},  R_j \leq \rlos\right) + \E\left( (1-p_j) \log_2(1+ SNR_j^{-\alpha}, R_j > \rlos) \right)\nonumber\\
    &=& \frac{\overline{\wtot}}{k \lu} \sum_{j=1}^k \Bigg(\E\left(\log_2(1+SNR_j), R_j \leq \rlos\right) + \int_{\rlos}^{\infty} g(r)\log_2(1+c r^{-\alpha})f_{R_j}(r)\dr\Bigg).
    \label{eq:LoS_channel}
\end{IEEEeqnarray}

Under the assumption that links fail independently, the outage probability becomes:
\begin{IEEEeqnarray}{rCl}
    \P(D_U = 0)&=& \prod_{j=1}^k \P(\text{link }i\text{ fails}) = \prod_{j=1}^k \E(p_j)\nonumber \\
    &=& \prod_{j=1}^k \int_{\rlos}^\infty g(r)f_{R_j}(r)\dr\nonumber \\
    &=& \prod_{j=1}^k \left(\P(R_j \geq \rlos) - \int_{\rlos}^\infty g(r) f_{R_j}(r)\dr \right) \nonumber \\
    &=& \prod_{j=1}^k \frac{\Gamma(j, \rlos^2\lbs \pi)}{\Gamma(j)} - \int_{\rlos}^\infty g(r) f_{R_j}(r)\dr.\label{eq:los_outage}
\end{IEEEeqnarray}

\begin{figure}[t] 
\centering
    \includegraphics[width=.4\textwidth]{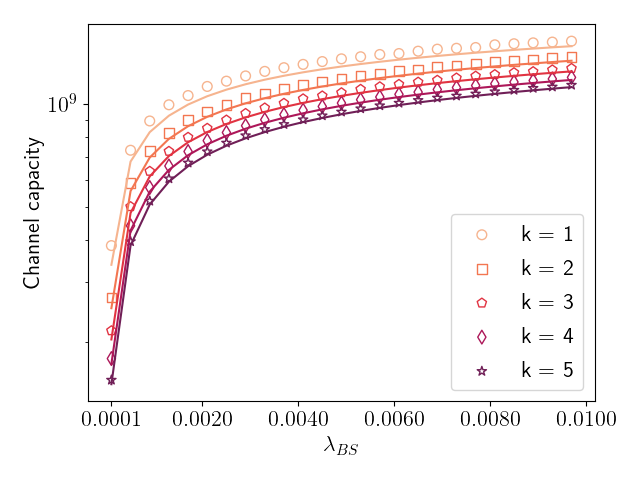}
    \caption{Impact of increasing BS density on the channel capacity. Simulated (markers) and calculated (line) expected channel capacity, $\lu = 0.1$, $\alpha = 2$ and $ c = 10^{3.5}$, on a 1500x1500m area, so $225,000$ users.}
    \label{fig:capacity_nofailure}
\end{figure}

\section{Performance Analysis}\label{sec:perf_analysis}
In this section, we provide a performance analysis of MC using system level Monte-Carlo simulations. Unless stated otherwise, we use the following parameters: $\lbs = 10^{-2}$, $\alpha=2$, $c = 10^{3.5}$ and $\lambda_{U}=0.1$, an area of 1500\,m$\times$1500\,m resulting in 225000 users. 

Fig.~\ref{fig:capacity_nofailure}  shows the channel capacity with increasing BS densities for different values of MC without failures.
The expected channel capacity in this plot is derived from Theorem \ref{thm:expectation_log_snr}. 
We have the following three observations. First, for all BS deployment densities, MC decreases the channel capacity, i.e., single connectivity $k=1$ achieves the highest capacity. 
This observation confirms the validity of Theorem \ref{th:decreasing_snr} which shows that MC decreases the channel capacity. 
Second, the relative performance difference between SC and MC is larger under sparse BS deployments, e.g., $\lambda_{BS}=10^{-4}$. 
That is, for sparse network deployments, increasing MC results in a more significant decrease in per-user throughput. However, in these networks, MC provides benefits in terms of lower outage probability.  
Third, while matching the trend, the analytical values of the channel capacity might deviate from the results obtained from the simulations. 
We attribute this deviation to  the assumption that the degree of a BS $D_{BS}$ and the distance to that BS $R_j$ are independent.
This assumption may not hold in general, as low-degree BSs most likely only serve a small area. Therefore, the distance from a user to that BS is also small. However, for $k \geq 2$, this dependence becomes less prominent as the areas that BSs serve become larger and more equally distributed \cite{stegehuis2021degree}.\\

Now, let us observe how failures affect the per-user throughput and outage probability. 
Figs.~\ref{fig:randomfailure}- \ref{fig:losfailure} depict the channel capacity and outage probability under different failures with/without re-allocation. 
In the figures, solid lines represent the analytical expressions of the channel capacities and outage probabilities. We use \eqref{eq:random_channel} and \eqref{eq:random_outage} for random failures, \eqref{eq:overload_failure} and \eqref{eq:overload_channel} for overload failures, \eqref{eq:distance_channel} and \eqref{eq:distance_outage} for  distance failures, and  \eqref{eq:LoS_channel} and \eqref{eq:los_outage} for  line-of-sight failures. 

\begin{figure*}[thb]
    \subcaptionbox{$\E\left(C^k_{sum}\right)$ - without re-allocation\label{sfig:random_cap}}{\includegraphics[width=.33\textwidth]{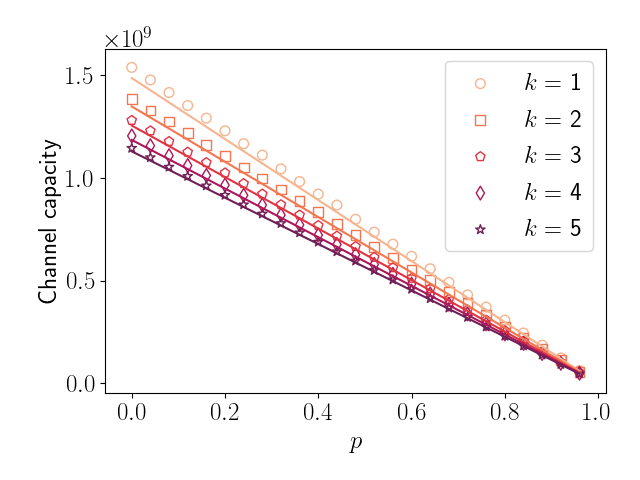}}\hfill
    \subcaptionbox{$\E\left(C^k_{sum}\right)$ - with re-allocation\label{sfig:random_cap_realloc}}{\includegraphics[width=.33\textwidth]{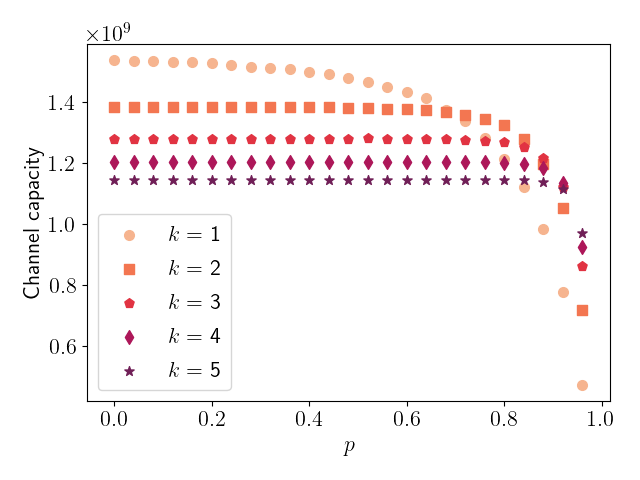}}\hfill
    \subcaptionbox{Outage probability\label{sfig:random_outage}}{\includegraphics[width=.33\textwidth]{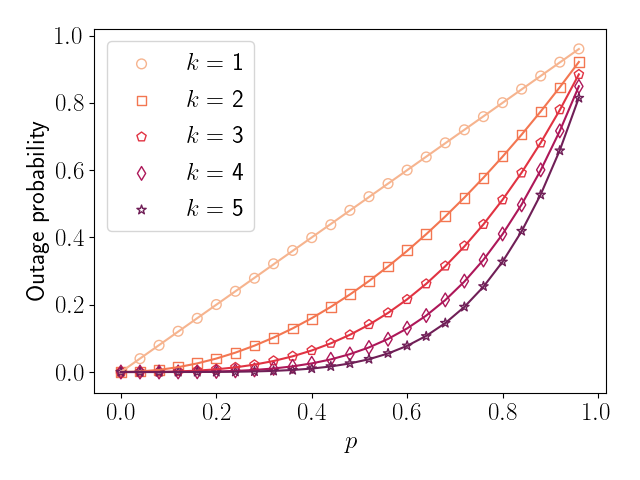}}\hfill
    \caption{Random failures. Impact of outage probability $p$ on simulated (markers) and calculated (line) channel capacity and outage probability. }
    \label{fig:randomfailure}
\end{figure*}

\begin{figure*}[thb]
    \subcaptionbox{$\E\left(C^k_{sum}\right)$ - without re-allocation\label{sfig:overload_cap}}{\includegraphics[width=.33\textwidth]{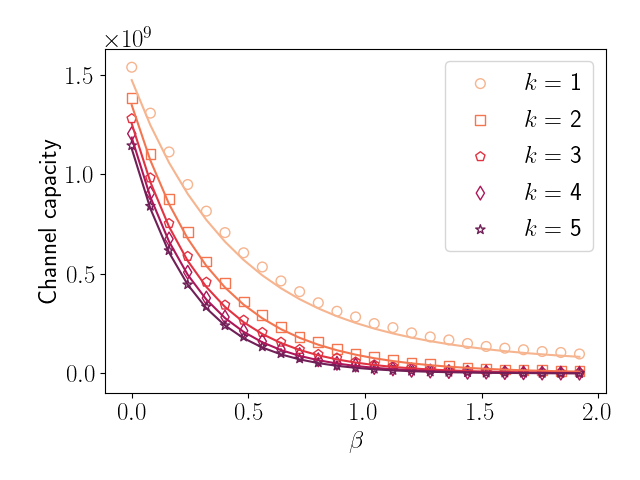}}\hfill
    \subcaptionbox{$\E\left(C^k_{sum}\right)$ - with re-allocation\label{sfig:overload_cap_realloc}}{\includegraphics[width=.33\textwidth]{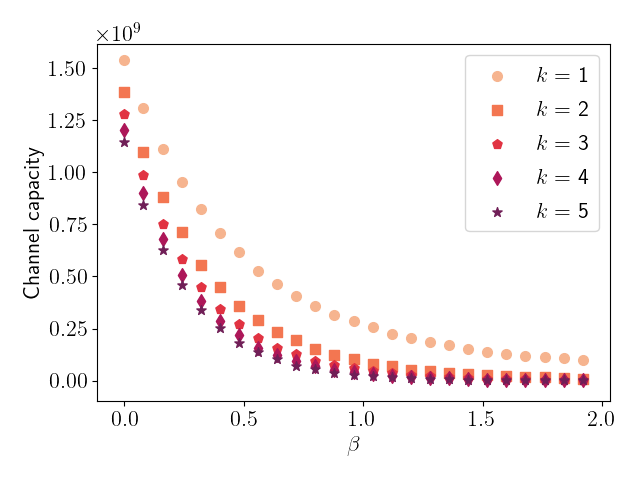}}\hfill
    \subcaptionbox{Outage probability\label{sfig:overload_outage}}{\includegraphics[width=.33\textwidth]{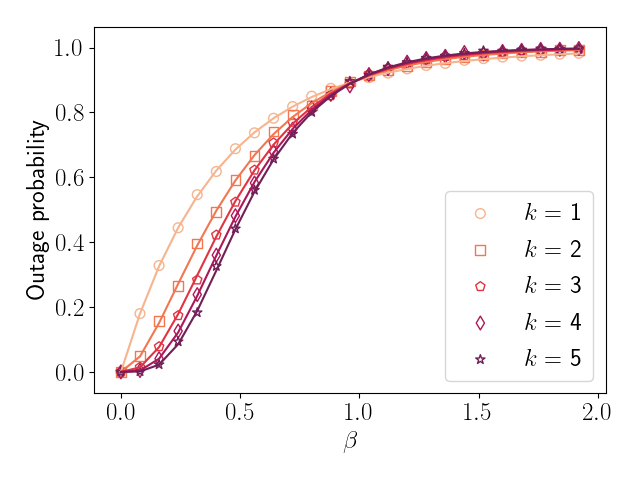}}\hfill
    \caption{Overload failures. Impact of parameter $\beta$ on simulated (markers) and calculated (line) channel capacity and outage probability. }
    \label{fig:overloadfailure}
\end{figure*}

\begin{figure*}[thb]
    \subcaptionbox{$\E\left(C^k_{sum}\right)$ - without re-allocation\label{sfig:distance_cap}}{\includegraphics[width=.33\textwidth]{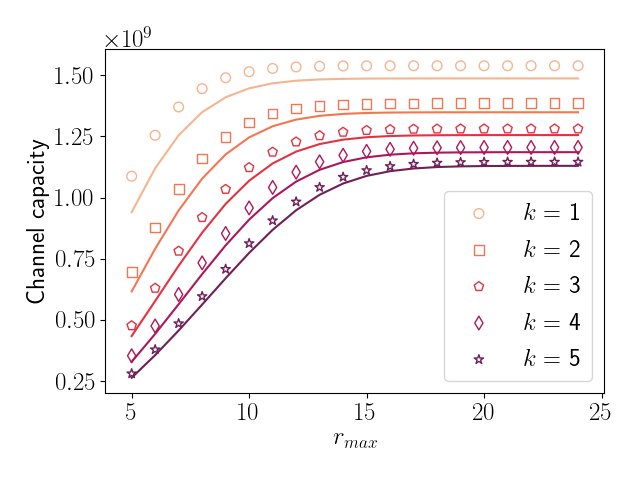}}\hfill
    \subcaptionbox{$\E\left(C^k_{sum}\right)$ - with re-allocation\label{sfig:distance_cap_realloc}}{\includegraphics[width=.33\textwidth]{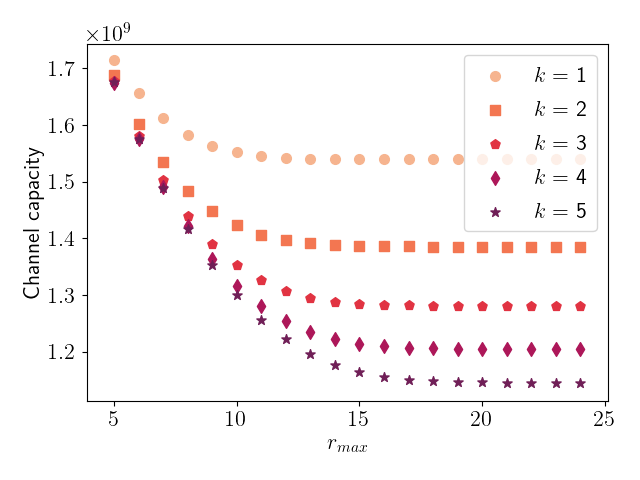}}\hfill
    \subcaptionbox{Outage probability\label{sfig:distance_outage}}{\includegraphics[width=.33\textwidth]{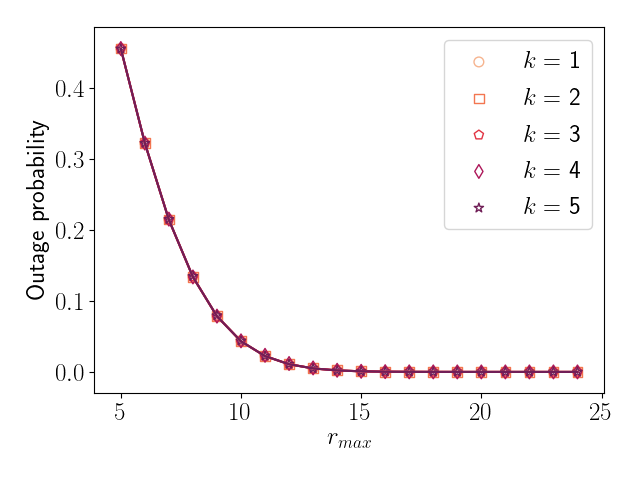}}\hfill
    \caption{Distance failures. Impact of parameter $r_{max}$ on simulated (markers) and calculated (line) channel capacity and outage probability. }
    \label{fig:distancefailure}
\end{figure*}

\begin{figure*}[thb]
    \subcaptionbox{$\E\left(C^k_{sum}\right)$ - without re-allocation\label{sfig:los_cap}}{\includegraphics[width=.33\textwidth]{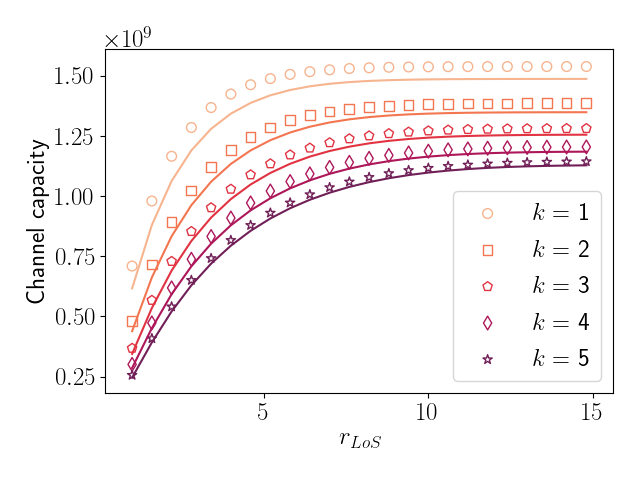}}\hfill
    \subcaptionbox{$\E\left(C^k_{sum}\right)$ - with re-allocation\label{sfig:los_cap_realloc}}{\includegraphics[width=.33\textwidth]{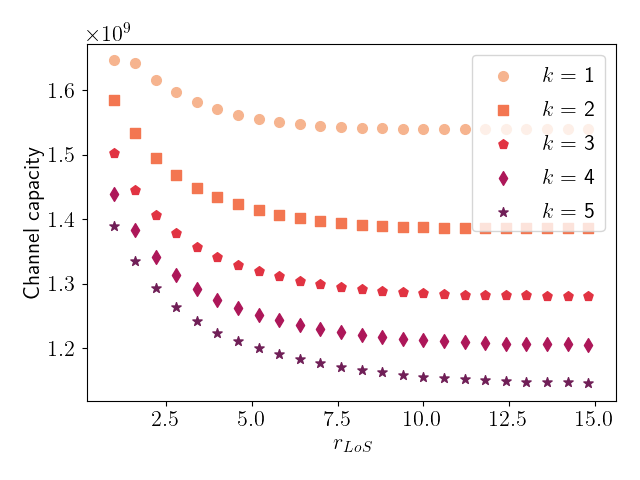}}\hfill
    \subcaptionbox{Outage probability\label{sfig:los_outage}}{\includegraphics[width=.33\textwidth]{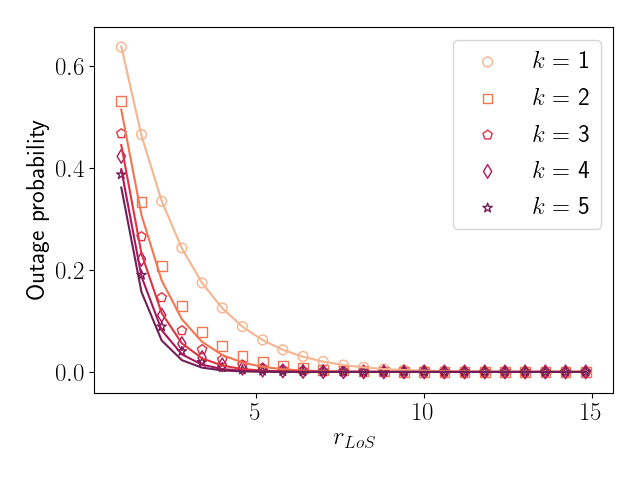}}\hfill
    \caption{Line-of-sight failures. Impact of parameter $\rlos$ on simulated (markers) and calculated (line) channel capacity and outage probability. }
    \label{fig:losfailure}
\end{figure*}

From Figs.\ref{sfig:random_cap}, \ref{sfig:overload_cap}, \ref{sfig:distance_cap}, \ref{sfig:los_cap}, we can see that even under failures, MC decreases the channel capacity of a user. This decrease is due to the fact that even if a user has more links that might survive after a failure, the additional links are with further away BSs. The SNR from more distant BSs are lower compared to the user's first link with the closest BS. Moreover, with increasing MC, bandwidth allocated to each user decreases, which results in lower per-user capacity~(Theorem \ref{th:decreasing_snr}).

A closer look to Fig.\ref{sfig:random_cap_realloc} shows that under high failure probabilities, the channel capacity is higher for larger values of $k$ when the bandwidth reserved for the failed links is re-allocated. 
In this case, SC leads to a lower channel capacity than MC for high $p$ since the bandwidth is divided amongst the links that survive.
%
Since the outage probability for SC is higher~(Fig.\ref{sfig:random_outage}), all users of a BS might experience network outage and therefore bandwidth re-allocation at this BS will not be possible. As a result, available bandwidth of this BS will remain idle, which leads to lower per-user throughput under SC as compared to MC. However, note that typical wireless networks have very robust links due to dense network deployment and advanced resource management schemes~(e.g., robust rate adaptation, error correction) yielding such high $p$ values very unlikely. 



Fig.\ref{fig:overloadfailure} depicts the impact of increasing $\beta$ under various degrees of MC. As we observe in Fig.\ref{sfig:overload_cap_realloc} and Fig. \ref{sfig:overload_cap}, channel capacity does not depend on the re-allocation policy as bandwidth re-allocation is irrelevant when a BS fails. Our key conclusion from the scenario with random failures holds also for overload failures. It is worth mentioning that we observe almost a perfect match between the analytical and simulated values of outage probability in Fig.\ref{sfig:random_outage} and Fig.\ref{sfig:overload_outage}. 
Additionally, considering the improvement in outage probability, the highest improvement from $k$ to $k+1$ connectivity is achieved for dual connectivity. This outage probability decrease enabled by a second link is especially visible for random failures for $p\sim 0.1$ in Fig.\ref{sfig:random_outage} and for $\beta\sim 0.1$ in Fig.\ref{sfig:overload_outage}. While higher $k$ decreases the outage probability further, the gain diminishes. Hence, one can argue that if the link failure probability is expected to be low, dual connectivity can facilitate the highest gains without resulting in a significant capacity loss or overhead to maintain multiple links~(e.g., scheduling coordination among the BSs). 

In Figs.\ref{fig:distancefailure} and \ref{fig:losfailure}, we observe the same impact of MC. However, comparing the capacity with and without re-allocation~(e.g., Fig.\ref{sfig:distance_cap} and Fig.\ref{sfig:distance_cap_realloc}), we observe a different trend. Here, when the capacity is not re-allocated, with increasing $r_{max}$ and $r_{LoS}$, channel capacity increases. In contrast, when re-allocation is possible, the channel capacity decreases. In terms of outage probability, as we see in Fig.\ref{sfig:distance_outage}, MC does not bring any benefits for distance related failures. When the failure is due to the loss of LoS, MC can present gains in terms of reliability but only $r_{LoS}$ smaller than 7.5 meters.


Moreover, Fig.\ref{fig:concentration_cksum} shows that the distribution of $C^k_{sum}$ over all users becomes more concentrated around the mean for larger values of $k$, as the coefficient of variation decreases. Thus, the average performance degrades, but \textit{fairness} increases implying that users maintain similar throughput. To illustrate this, we calculated Jain's fairness index \cite{jain1984quantitative} of the channel capacity as given Figure \ref{fig:concentration_cksum}. This fairness index increases for higher degrees of MC, as for $k = 1$, the fairness is $0.52$ and for $k = 5$ this index is $0.89$. Furthermore, the channel capacity does not decrease for all users: fewer users get a channel capacity of 1 Gbps or less under multi-connected networks in comparison to single-connected networks. Thus, MC degrades the average channel capacity, but increases the channel capacity for the poorly-connected users~(e.g., with lower achieved capacity at the cell-edge).


Finally, we show the minimum and maximum capacity loss between $k = 1$ (SC) and $k > 1$ with and without re-allocation in Table~\ref{tab:capacity_loss}. The minimum and maximum loss is defined as the decrease in throughput in comparison to the SC case. 
From Table~\ref{tab:capacity_loss},  it is clear that the capacity loss becomes larger when the network fails in comparison to no failure, 
while it becomes smaller when bandwidth re-allocation is possible. This result can be explained by the fact that the surviving links will be allocated also the bandwidth of the failed links. 
Notice the negative capacity loss values for random failures in Table~\ref{tab:capacity_loss}. These cases correspond to the scenarios with very high $p$ as explained in Fig.\ref{sfig:random_cap_realloc}.

\renewcommand{\arraystretch}{1.3}

\begin{figure*}
  \begin{minipage}{.55\textwidth}
  \resizebox{\textwidth}{!}{
    \begin{tabular}{|l|cccc|} \hline 
        \textbf{Degree of MC} & \multicolumn{1}{c}{$k = 2$} & \multicolumn{1}{c}{$k = 3$} & \multicolumn{1}{c}{$k = 4$} & \multicolumn{1}{c|}{$k = 5$} \\ \hline
        \textbf{No failure} & $10.0\%$         & $16.8 \%$       & $21.7\%$        & $25.6 \%$       \\
        \hline
        \multicolumn{5}{|c|}{\textbf{Without re-allocation}} \\\hline
        \textbf{Random}                        & $8.7 / 10.8\%$  & $15.5 / 17.4\%$ & $20.7 / 22.2\%$ & $24.1 / 26.1\%$ \\
        \textbf{Overload}                      & $10.0 / 89.0\%$  & $16.8 / 96.3\%$ & $21.7 / 98.7\%$ & $25.6 / 99.2\%$ \\
        \textbf{Distance}                      & $10.0 / 36.0\%$  & $16.8 - 56.2\%$ & $21.7 / 67.5\%$ & $25.6 / 74.3\%$ \\
        \textbf{LoS}                 & $10.0 / 32.6\%$ & $16.8 / 48.4\%$ & $21.7 / 57.9\%$ & $25.6 / 63.9\%$ \\\hline
        \multicolumn{5}{|c|}{\textbf{With re-allocation}} \\\hline
        \textbf{Random}                        & $-52.5 / 10.0\%$  & $-82.8 / 16.8\%$ & $-96.0 / 21.7\%$ & $-105.9 / 25.6\%$ \\
        \textbf{Overload}                      & $10.0 / 89.0\%$  & $16.8 / 96.3\%$ & $21.7 / 98.7\%$ & $25.6 / 99.2\%$ \\
        \textbf{Distance}                      & $1.5 / 10.0\%$  & $2.1 / 16.8\%$ & $2.3 / 21.7\%$ & $2.3 / 25.6\%$ \\
        \textbf{LoS}                 & $3.9 / 10.0\%$ & $8.8 / 16.8\%$ & $12.6 / 21.7\%$ & $15.5 / 25.6\%$ \\\hline
    \end{tabular}%
    }
    \captionof{table}{Minimum and maximum loss in capacity in relation to $k = 1$ for different failures without re-allocation, for $\lbs = 10^{-2}$.}
    \label{tab:capacity_loss}
  \end{minipage} \quad
  \begin{minipage}{.40\textwidth}
    \includegraphics[width = \linewidth]{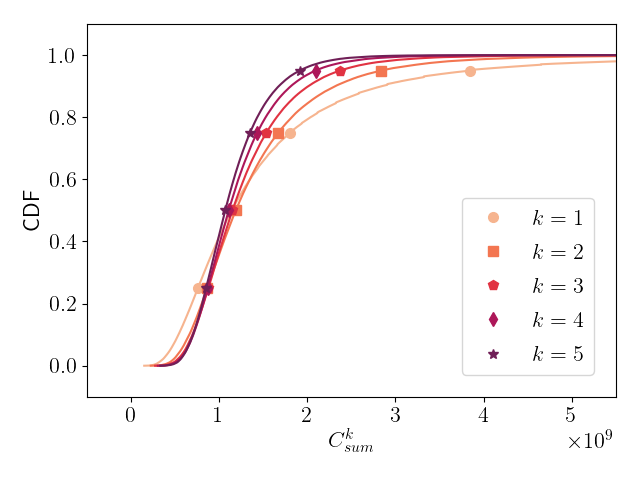}
    \caption{Simulated cumulative distribution of $C^k_{sum}$ for $\lbs = 10^{-2}$, $\lu = 0.1$, $\alpha = 2$ and $c = 10^{3.5}$.}
    \label{fig:concentration_cksum}
  \end{minipage}
\end{figure*}

%% file: tikz-failure-models.tex
\usetikzlibrary{arrows, arrows.meta}
\begin{figure*}[h!]
\centering
    \subcaptionbox{Random failure $p = 0.25$ \label{tikz:random_failure}}{
    \resizebox{0.31\textwidth}{!}{%
    \begin{tikzpicture}
        \node[inner sep = -2pt] (BS1) at (1,0.5) {\includegraphics[scale = 0.05]{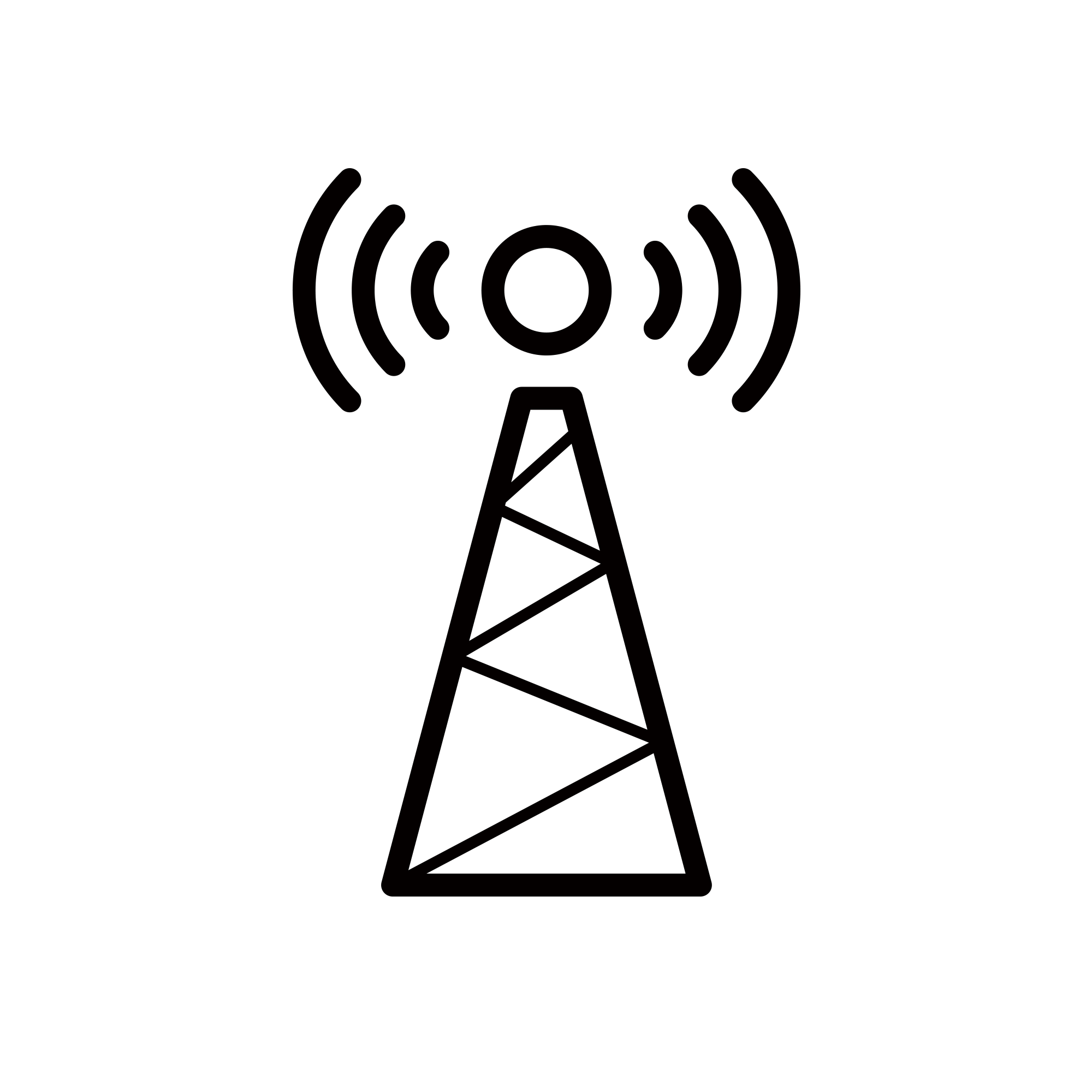}};
        \node[inner sep = -2pt] (BS2) at (-2.5,3) {\includegraphics[scale = 0.05]{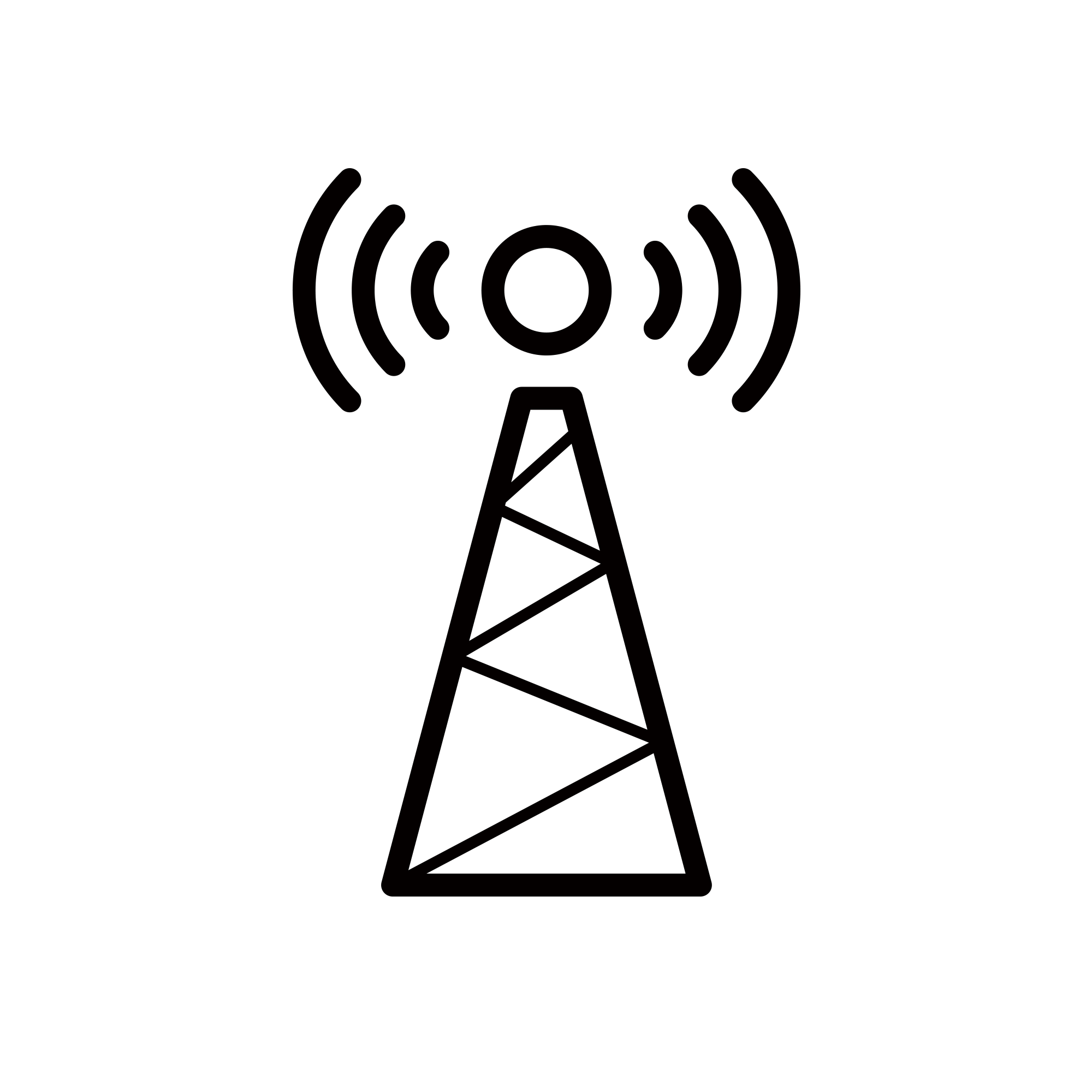}};
        \node[inner sep = -2pt] (BS3) at (3,-1.5) {\includegraphics[scale = 0.05]{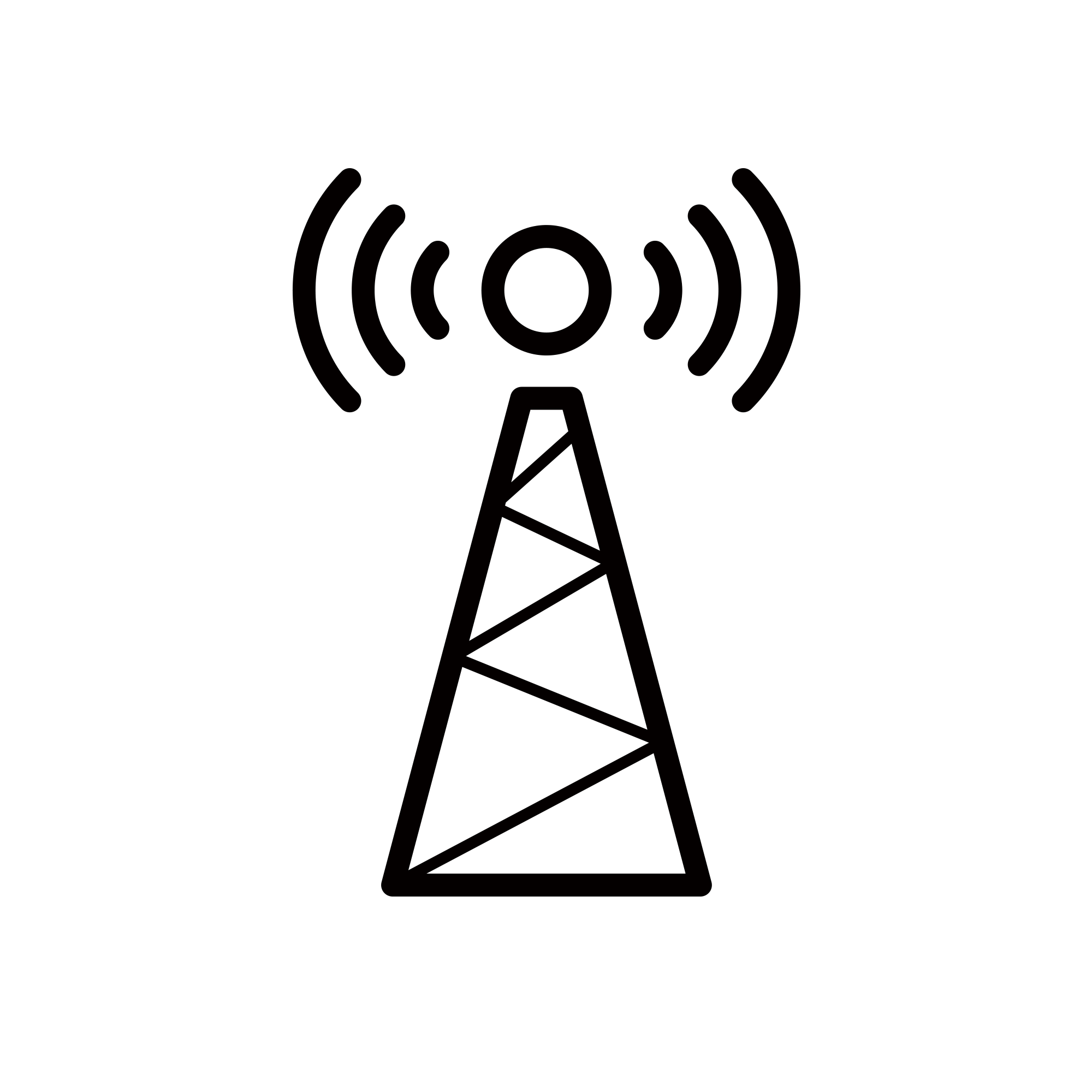}};
        \node[inner sep = -2pt] (BS4) at (-2,-2) {\includegraphics[scale = 0.05]{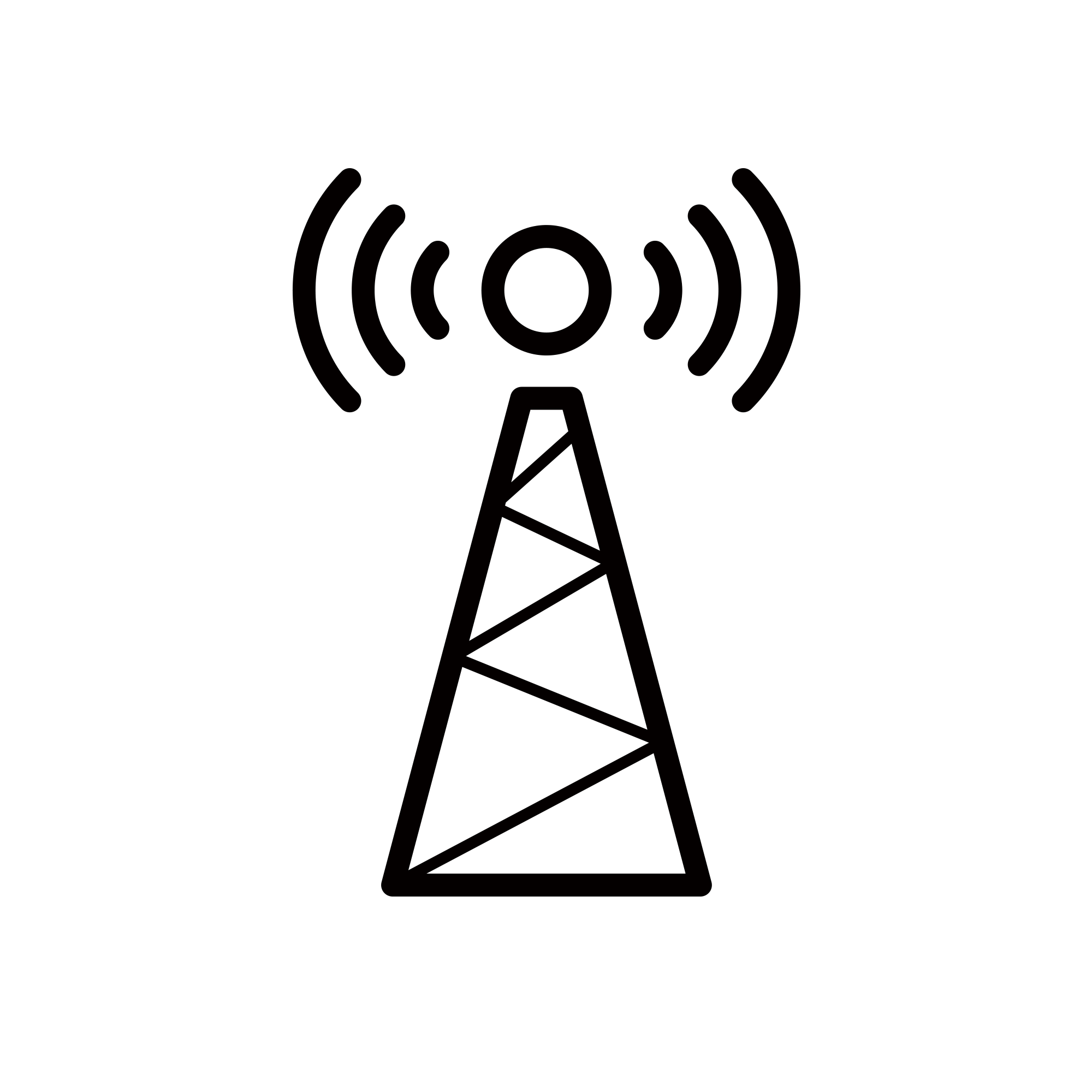}};
        
        \node[shape = circle, fill, color = green, scale =0.5] (U1) at (-2,1.5) {};
        \node[shape = circle, fill, color = green, scale =0.5]  (U2) at (-1,2) {};
        \node[shape = circle, fill, color = green, scale =0.5] (U3) at (2,2.5) {};
        \node[shape = circle, fill, color = red, scale =0.5]  (U4) at (0,3.5) {};
        \node[shape = circle, fill, color = green, scale =0.5] (U5) at (-2.5,0.5) {};
        \node[shape = circle, fill, color = green, scale =0.5] (U6) at (-1.5,-1) {};
        \node[shape = circle, fill, color = yellow, scale =0.5] (U7) at (2,-3) {};
        \node[shape = circle, fill, color = red, scale =0.5]  (U8) at (-1.5,-3) {};
        \node[shape = circle, fill, color = yellow, scale =0.5] (U9) at (2,-1) {};
        \node[shape = circle, fill, color = green, scale =0.5] (U10) at (4,-2) {};
        \node[shape = circle, fill, color = yellow, scale =0.5] (U11) at (2,0) {};
        \node[shape = circle, fill, color = green, scale =0.5] (U12) at (3.5,0.5) {};
        \node[shape = circle, fill, color = yellow, scale =0.5] (U13) at (-3.5,-1.5) {};
        \node[shape = circle, fill, color = green, scale =0.5] (U14) at (-4,-0.5) {};
        \node[shape = circle, fill, color = green, scale =0.5] (U15) at (4.5,1.5) {};
        \node[shape = circle, fill, color = green, scale =0.5] (U16) at (-3.5,3) {};
        \node[shape = circle, fill, color = green, scale =0.5] (U17) at (-1,0.5) {};
        \node[shape = circle, fill, color = yellow, scale =0.5] (U18) at (-0.5,0) {};
        \node[shape = circle, fill, color = green, scale =0.5] (U19) at (1,-2) {};
        
        \draw[color = lightgray, thick] (U16) -- (BS2);
        \draw[color = lightgray, thick] (U1) -- (BS2);
        \draw[color = lightgray, thick] (U2) -- (BS2);
        \draw[color = red, dotted, thick] (U4) -- (BS2);
        \draw[color = lightgray, thick] (U5) -- (BS2);
        
        \draw[color = lightgray, thick] (U14) -- (BS4);
        \draw[color = lightgray, thick] (U13) -- (BS4);
        \draw[color = red, dotted, thick] (U8) -- (BS4);
        \draw[color = lightgray, thick] (U6) -- (BS4);
        
        \draw[color = lightgray, thick] (U18) -- (BS1);
        \draw[color = lightgray, thick] (U17) -- (BS1);
        \draw[color = red, dotted, thick] (U11) -- (BS1);
        \draw[color = lightgray, thick] (U3) -- (BS1);
        
        \draw[color = lightgray, thick] (U15) -- (BS3);
        \draw[color = lightgray, thick] (U12) -- (BS3);
        \draw[color = red, dotted, thick] (U9) -- (BS3);
        \draw[color = lightgray, thick] (U7) -- (BS3);
        \draw[color = lightgray, thick] (U10) -- (BS3);
        \draw[color = lightgray, thick] (U19) -- (BS3);
        
        \draw [color = red, dotted, thick](U13) -- (BS2);
        \draw [color = lightgray, thick](U14) -- (BS2);
        \draw [color = lightgray, thick](U17) -- (BS2);
        
        \draw [color = lightgray, thick](U5) -- (BS4);
        \draw [color = red, dotted, thick](U18) -- (BS4);
        \draw [color = lightgray, thick](U19) -- (BS1);
        
        \draw [color = lightgray, thick](U1) -- (BS1);
        \draw [color = lightgray, thick](U16) -- (BS1);
        \draw [color = red, dotted, thick](U4) -- (BS1);
        \draw [color = lightgray, thick](U2) -- (BS1);
        \draw [color = lightgray, thick](U9) -- (BS1);
        \draw [color = lightgray, thick](U6) -- (BS1);
        \draw [color = red, dotted, thick](U8) -- (BS1);
        \draw [color = lightgray, thick](U12) -- (BS1);
        \draw [color = lightgray, thick](U15) -- (BS1);
        \draw [color = lightgray, thick](U10) -- (BS1);
        \draw [color = red, dotted, thick](U7) -- (BS1);
        
        \draw [color = lightgray, thick](U3) -- (BS3);
        \draw [color = lightgray, thick](U11) -- (BS3);
        
    \end{tikzpicture}%
    }}
    \subcaptionbox{Overload failure \label{tikz:overload_failure}}{
    \resizebox{0.31\textwidth}{!}{%
    \begin{tikzpicture}
        \node[inner sep = -2pt] (BS1) at (1,0.5) {\includegraphics[scale = 0.05]{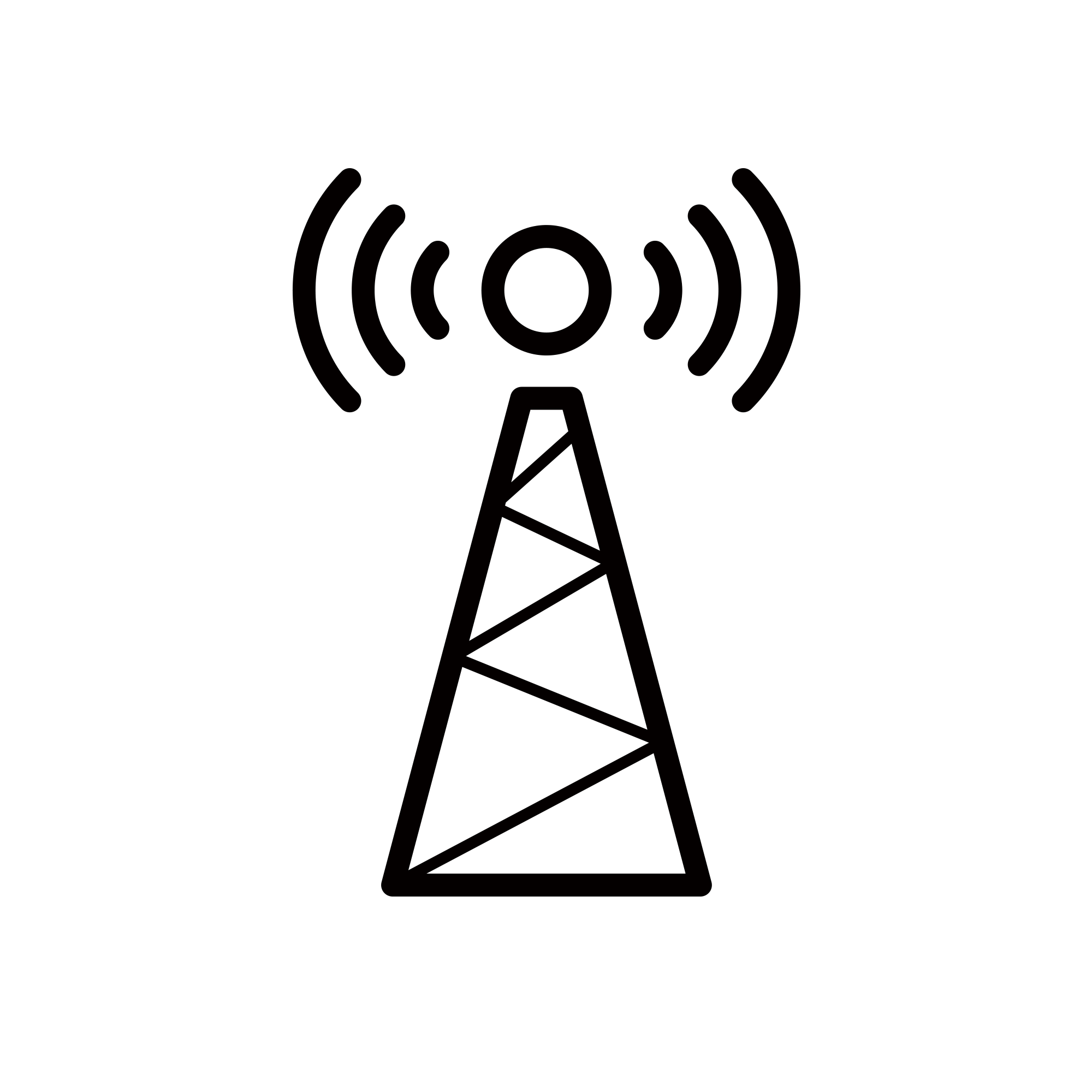}};
        \node[inner sep = -2pt] (BS2) at (-2.5,3) {\includegraphics[scale = 0.05]{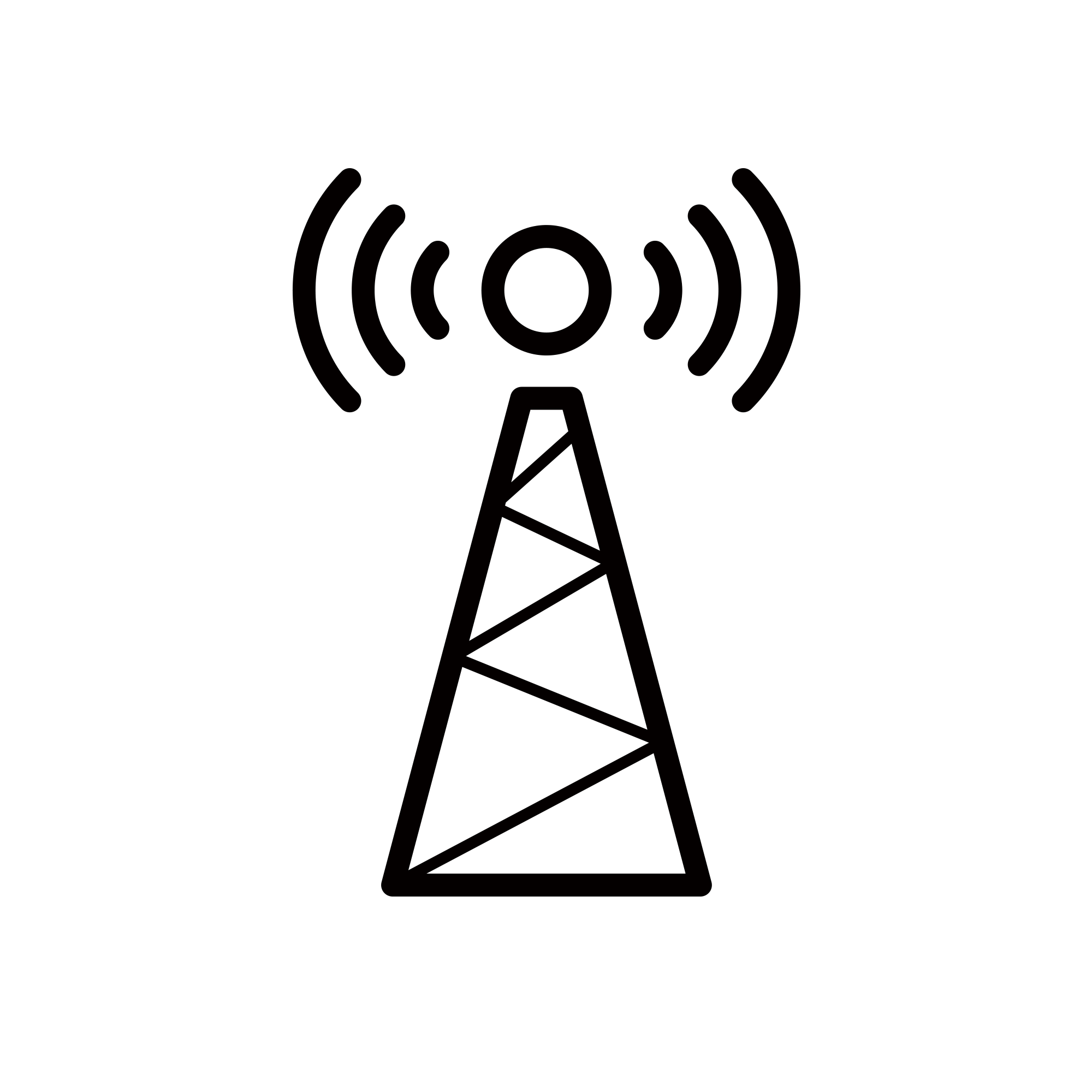}};
        \node[inner sep = -2pt] (BS3) at (3,-1.5) {\includegraphics[scale = 0.05]{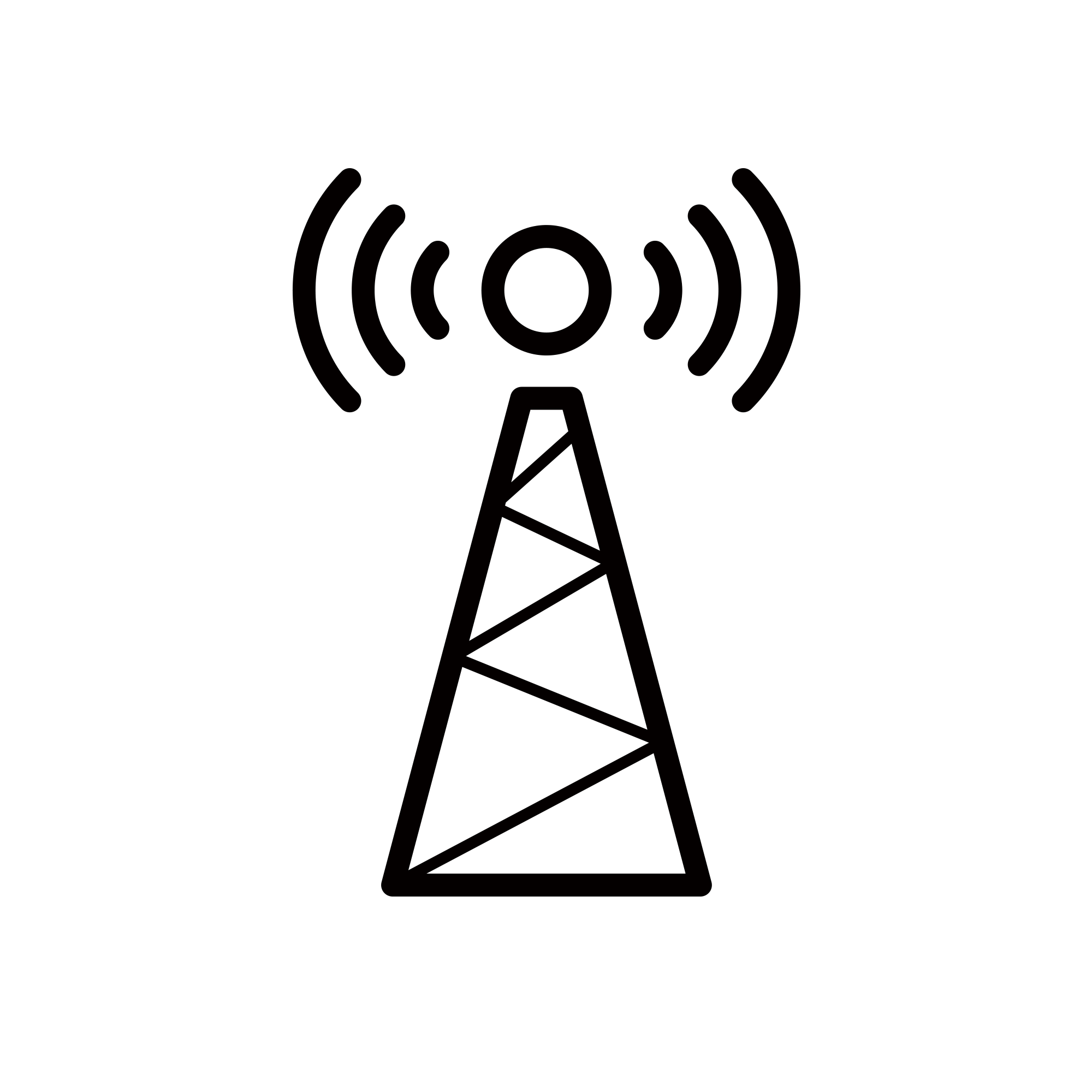}};
        \node[inner sep = -2pt] (BS4) at (-2,-2) {\includegraphics[scale = 0.05]{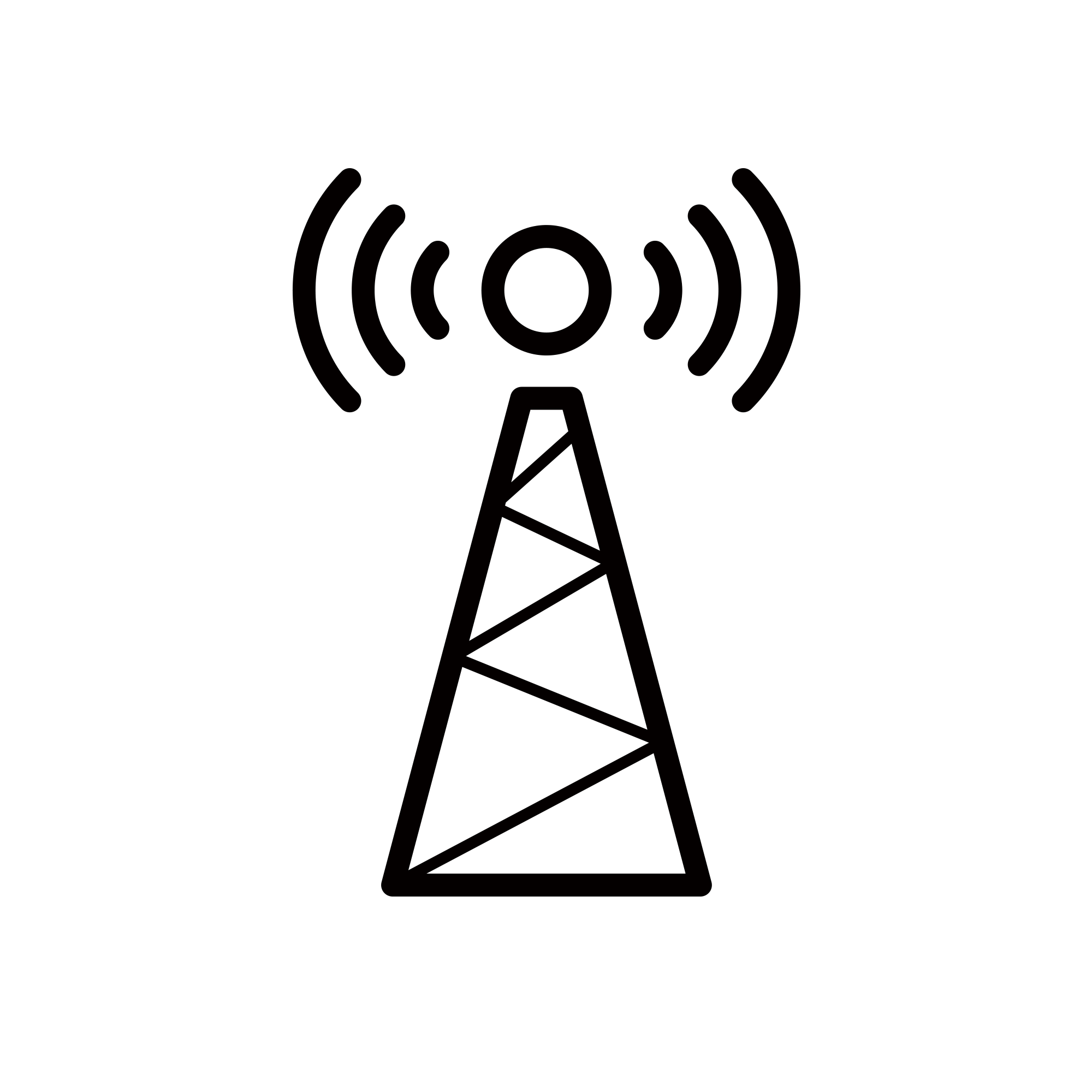}};
        
        \draw[red, thick] (1.5,1) -- (0.5,0);
        \draw[red, thick] (0.5,1) -- (1.5,0);
        
        \node[right] at (0.43,2) {\footnotesize $D_{BS} \geq 9$};
        \draw[-Latex]  plot[smooth, tension=.7] coordinates {(1.05,0.9) (1.1,1.3) (1.15,1.7)};
        
        \node[shape = circle, fill, color = yellow, scale =0.5] (U1) at (-2,1.5) {};
        \node[shape = circle, fill, color = yellow, scale =0.5]  (U2) at (-1,2) {};
        \node[shape = circle, fill, color = yellow, scale =0.5] (U3) at (2,2.5) {};
        \node[shape = circle, fill, color = yellow, scale =0.5]  (U4) at (0,3.5) {};
        \node[shape = circle, fill, color = green, scale =0.5] (U5) at (-2.5,0.5) {};
        \node[shape = circle, fill, color = yellow, scale =0.5] (U6) at (-1.5,-1) {};
        \node[shape = circle, fill, color = yellow, scale =0.5] (U7) at (2,-3) {};
        \node[shape = circle, fill, color = yellow, scale =0.5]  (U8) at (-1.5,-3) {};
        \node[shape = circle, fill, color = yellow, scale =0.5] (U9) at (2,-1) {};
        \node[shape = circle, fill, color = yellow, scale =0.5] (U10) at (4,-2) {};
        \node[shape = circle, fill, color = yellow, scale =0.5] (U11) at (2,0) {};
        \node[shape = circle, fill, color = yellow, scale =0.5] (U12) at (3.5,0.5) {};
        \node[shape = circle, fill, color = green, scale =0.5] (U13) at (-3.5,-1.5) {};
        \node[shape = circle, fill, color = green, scale =0.5] (U14) at (-4,-0.5) {};
        \node[shape = circle, fill, color = yellow, scale =0.5] (U15) at (4.5,1.5) {};
        \node[shape = circle, fill, color = yellow, scale =0.5] (U16) at (-3.5,3) {};
        \node[shape = circle, fill, color = yellow, scale =0.5] (U17) at (-1,0.5) {};
        \node[shape = circle, fill, color = yellow, scale =0.5] (U18) at (-0.5,0) {};
        \node[shape = circle, fill, color = yellow, scale =0.5] (U19) at (1,-2) {};
        
        \draw[color = lightgray, thick] (U16) -- (BS2); 
        \draw[color = lightgray, thick] (U1) -- (BS2);
        \draw[color = lightgray, thick] (U2) -- (BS2);
        \draw[color = lightgray, thick] (U4) -- (BS2);
        \draw[color = lightgray, thick] (U5) -- (BS2);
        
        \draw[color = lightgray, thick] (U14) -- (BS4);
        \draw[color = lightgray, thick] (U13) -- (BS4);
        \draw[color = lightgray, thick] (U8) -- (BS4);
        \draw[color = lightgray, thick] (U6) -- (BS4);
        
        \draw[color = red, dotted, thick] (U18) -- (BS1);
        \draw[color = red, dotted, thick] (U17) -- (BS1);
        \draw[color = red, dotted, thick] (U11) -- (BS1);
        \draw[color = red, dotted, thick] (U3) -- (BS1);
        
        \draw[color = lightgray, thick] (U15) -- (BS3);
        \draw[color = lightgray, thick] (U12) -- (BS3);
        \draw[color = lightgray, thick] (U9) -- (BS3);
        \draw[color = lightgray, thick] (U7) -- (BS3);
        \draw[color = lightgray, thick] (U10) -- (BS3);
        \draw[color = lightgray, thick] (U19) -- (BS3);
        
        \draw [color = lightgray, thick](U13) -- (BS2);
        \draw [color = lightgray, thick](U14) -- (BS2);
        \draw [color = lightgray, thick](U17) -- (BS2);
        
        \draw [color = lightgray, thick](U5) -- (BS4);
        \draw [color = lightgray, thick](U18) -- (BS4);
        
        \draw [color = red, dotted, thick](U19) -- (BS1);
        \draw [color = red, dotted, thick](U1) -- (BS1);
        \draw [color = red, dotted, thick](U16) -- (BS1);
        \draw [color = red, dotted, thick](U4) -- (BS1);
        \draw [color = red, dotted, thick](U2) -- (BS1);
        \draw [color = red, dotted, thick](U9) -- (BS1);
        \draw [color = red, dotted, thick](U6) -- (BS1);
        \draw [color = red, dotted, thick](U8) -- (BS1);
        \draw [color = red, dotted, thick](U12) -- (BS1);
        \draw [color = red, dotted, thick](U15) -- (BS1);
        \draw [color = red, dotted, thick](U10) -- (BS1);
        \draw [color = red, dotted, thick](U7) -- (BS1);
        
        \draw [color = lightgray, thick](U3) -- (BS3);
        \draw [color = lightgray, thick](U11) -- (BS3);
        
        \end{tikzpicture}%
        }}
        \subcaptionbox{Distance failure, $r_{max} = 300m$ \label{tikz:det_distance_failure}}{
        \resizebox{0.31\textwidth}{!}{%
        \begin{tikzpicture}
        
        \node[inner sep = -2pt] (BS1) at (1,0.5) {\includegraphics[scale = 0.05]{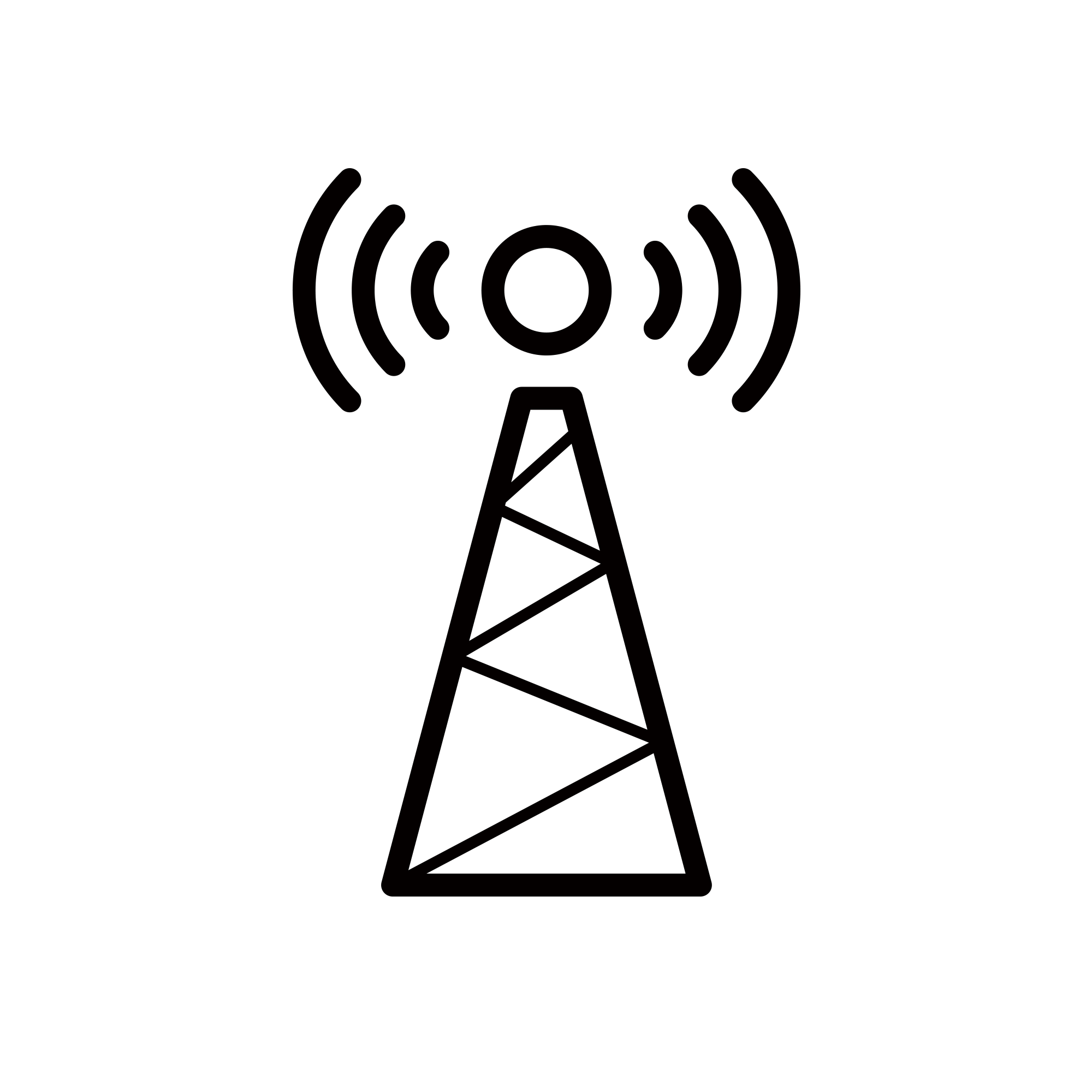}};
        \node[inner sep = -2pt] (BS2) at (-2.5,3) {\includegraphics[scale = 0.05]{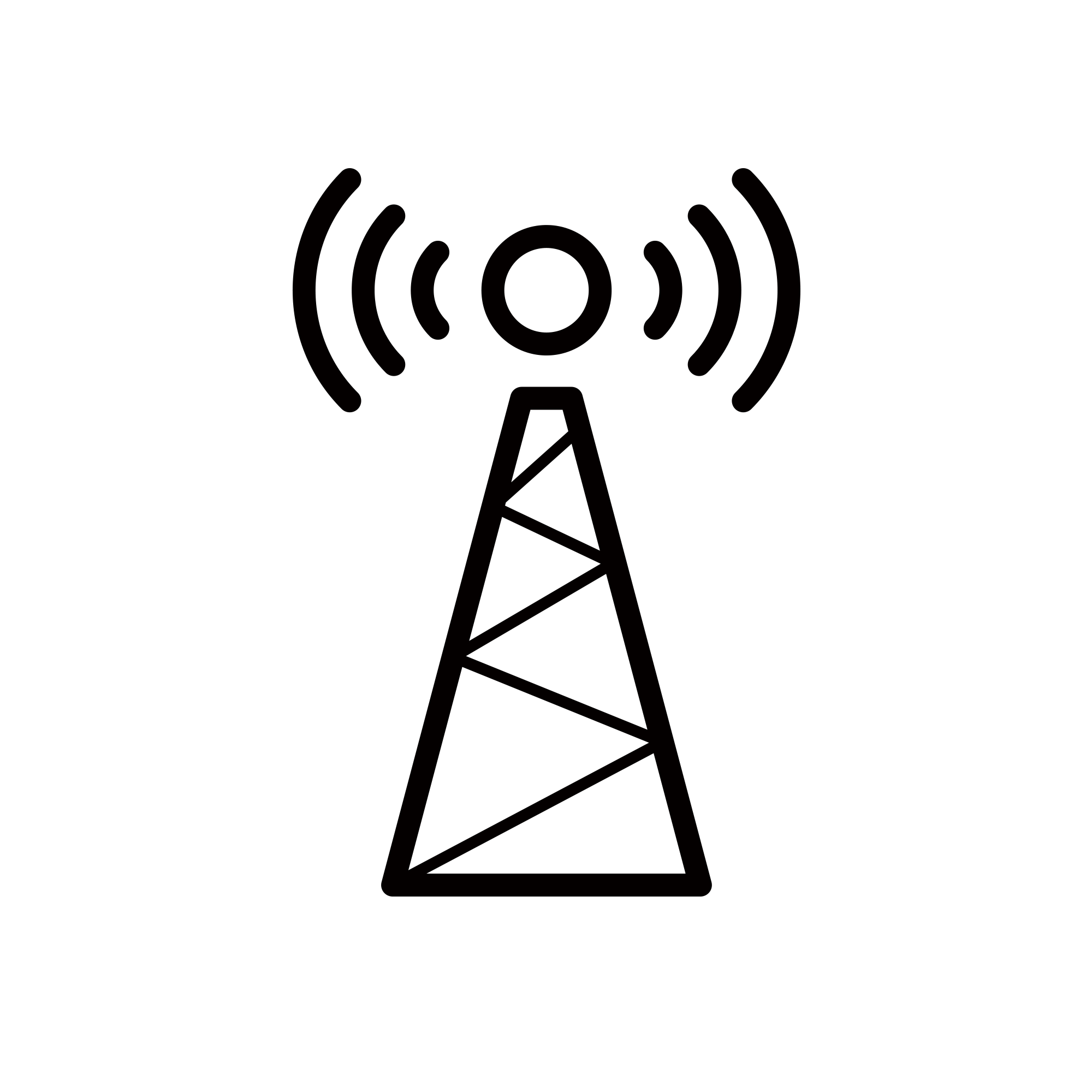}};
        \node[inner sep = -2pt] (BS3) at (3,-1.5) {\includegraphics[scale = 0.05]{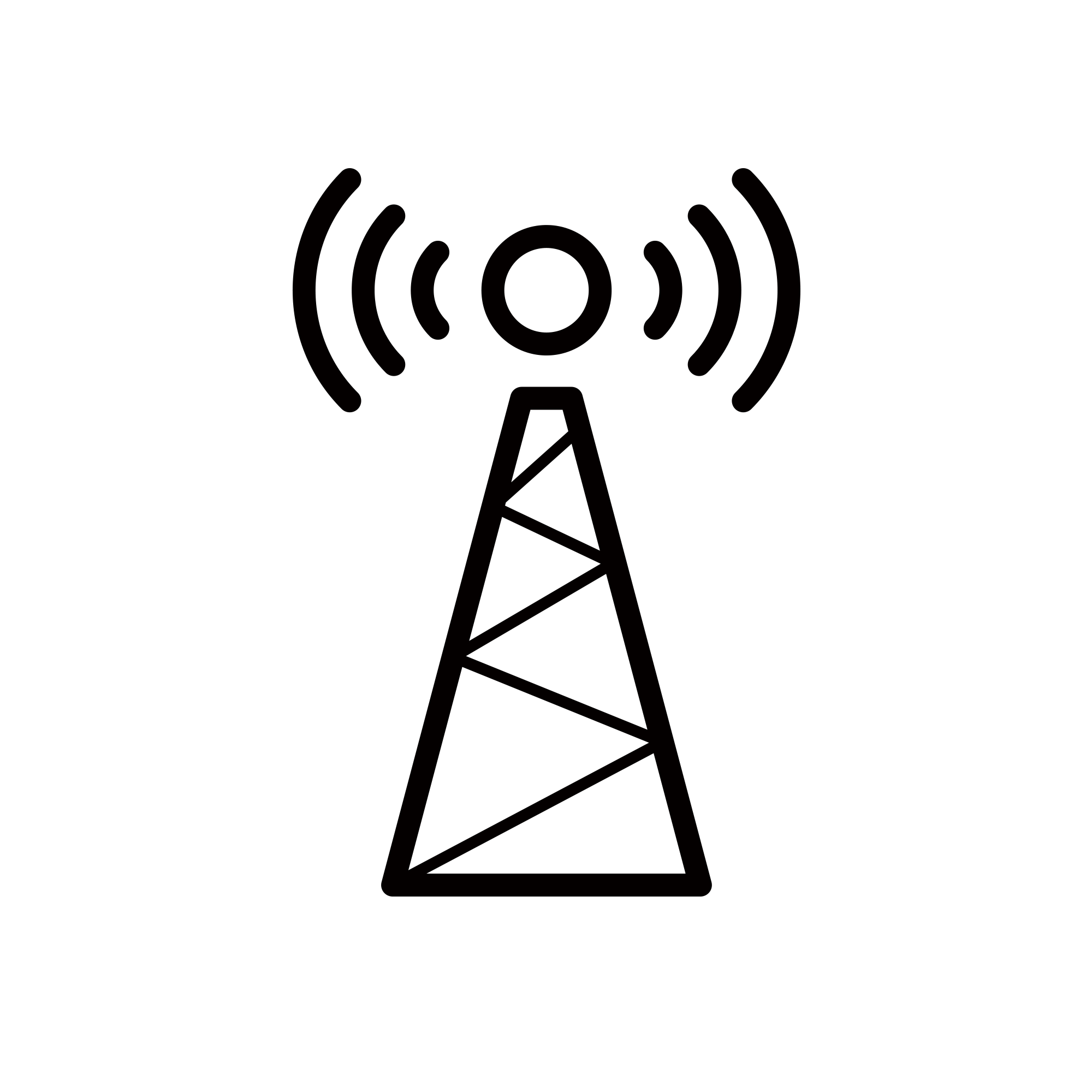}};
        \node[inner sep = -2pt] (BS4) at (-2,-2) {\includegraphics[scale = 0.05]{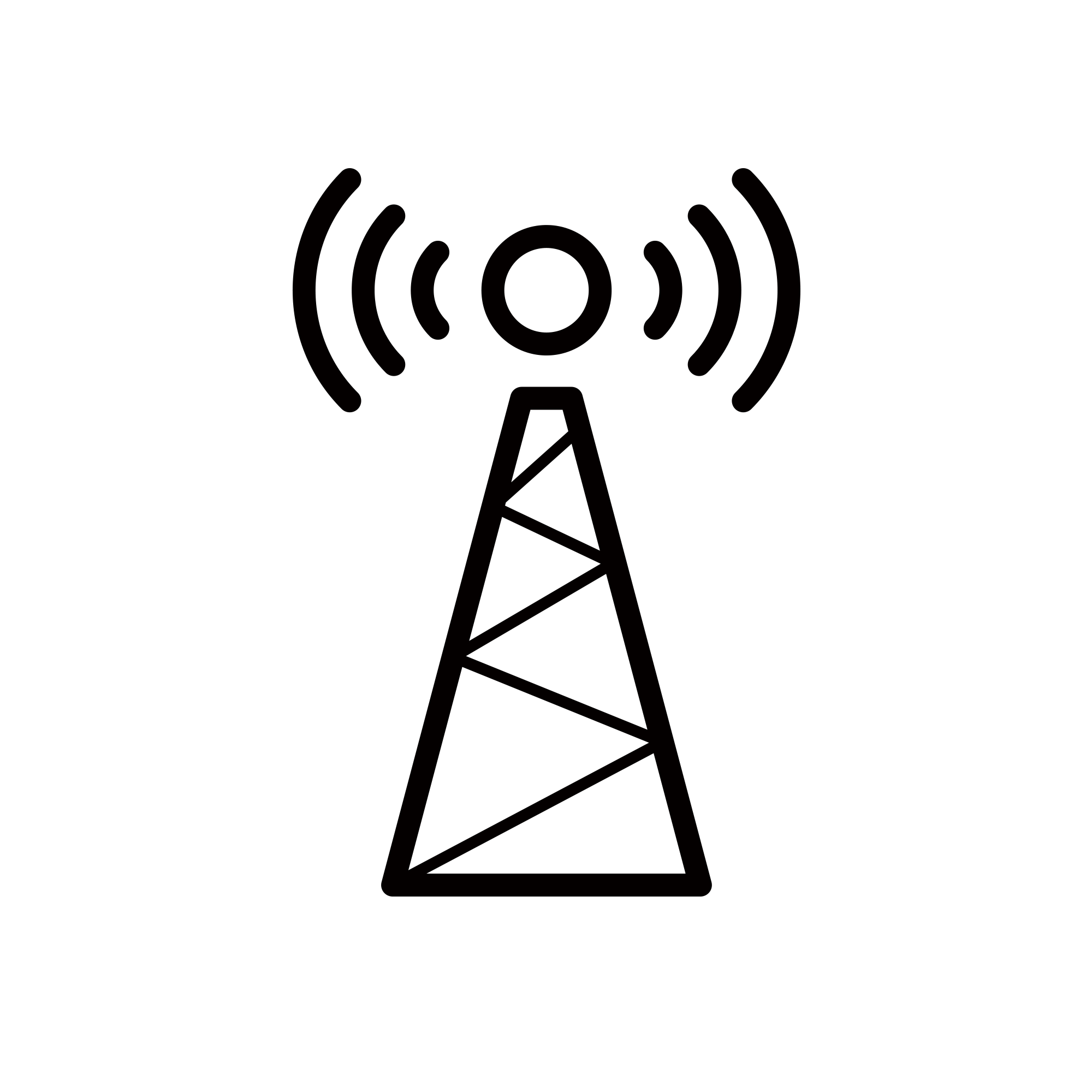}};
        
        \node[shape = circle, fill, color = yellow, scale =0.5] (U1) at (-2,1.5) {};
        \node[shape = circle, fill, color = green, scale =0.5]  (U2) at (-1,2) {};
        \node[shape = circle, fill, color = yellow, scale =0.5] (U3) at (2,2.5) {};
        \node[shape = circle, fill, color = yellow, scale =0.5]  (U4) at (0,3.5) {};
        \node[shape = circle, fill, color = green, scale =0.5] (U5) at (-2.5,0.5) {};
        \node[shape = circle, fill, color = green, scale =0.5] (U6) at (-1.5,-1) {};
        \node[shape = circle, fill, color = yellow, scale =0.5] (U7) at (2,-3) {};
        \node[shape = circle, fill, color = yellow, scale =0.5]  (U8) at (-1.5,-3) {};
        \node[shape = circle, fill, color = green, scale =0.5] (U9) at (2,-1) {};
        \node[shape = circle, fill, color = yellow, scale =0.5] (U10) at (4,-2) {};
        \node[shape = circle, fill, color = green, scale =0.5] (U11) at (2,0) {};
        \node[shape = circle, fill, color = green, scale =0.5] (U12) at (3.5,0.5) {};
        \node[shape = circle, fill, color = yellow, scale =0.5] (U13) at (-3.5,-1.5) {};
        \node[shape = circle, fill, color = yellow, scale =0.5] (U14) at (-4,-0.5) {};
        \node[shape = circle, fill, color = red, scale =0.5] (U15) at (4.5,1.5) {};
        \node[shape = circle, fill, color = yellow, scale =0.5] (U16) at (-3.5,3) {};
        \node[shape = circle, fill, color = green, scale =0.5] (U17) at (-1,0.5) {};
        \node[shape = circle, fill, color = green, scale =0.5] (U18) at (-0.5,0) {};
        \node[shape = circle, fill, color = green, scale =0.5] (U19) at (1,-2) {};
        
        \draw[color = lightgray, thick] (U16) -- (BS2); 
        \draw[color = lightgray, thick] (U1) -- (BS2);
        \draw[color = lightgray, thick] (U2) -- (BS2);
        \draw[color = lightgray, thick] (U4) -- (BS2);
        \draw[color = lightgray, thick] (U5) -- (BS2);
        
        \draw[color = lightgray, thick] (U14) -- (BS4);
        \draw[color = lightgray, thick] (U13) -- (BS4);
        \draw[color = lightgray, thick] (U8) -- (BS4);
        \draw[color = lightgray, thick] (U6) -- (BS4);
        
        \draw[color = lightgray, thick] (U18) -- (BS1);
        \draw[color = lightgray, thick] (U17) -- (BS1);
        \draw[color = lightgray, thick] (U11) -- (BS1);
        \draw[color = lightgray, thick] (U3) -- (BS1);
        
        \draw[color =  red, dotted, thick] (U15) -- (BS3);
        \draw[color = lightgray, thick] (U12) -- (BS3);
        \draw[color = lightgray, thick] (U9) -- (BS3);
        \draw[color = lightgray, thick] (U7) -- (BS3);
        \draw[color = lightgray, thick] (U10) -- (BS3);
        \draw[color = lightgray, thick] (U19) -- (BS3);
        
        \draw [color =  red, dotted, thick](U13) -- (BS2);
        \draw [color =  red, dotted, thick](U14) -- (BS2);
        \draw [color = lightgray, thick](U17) -- (BS2);
        
        \draw [color = lightgray, thick](U5) -- (BS4);
        \draw [color = lightgray, thick](U18) -- (BS4);
        
        \draw [color = lightgray, thick](U19) -- (BS1);
        \draw [color = red, dotted, thick](U1) -- (BS1);
        \draw [color =  red, dotted, thick](U16) -- (BS1);
        \draw [color =  red, dotted, thick](U4) -- (BS1);
        \draw [color = lightgray, thick](U2) -- (BS1);
        \draw [color = lightgray, thick](U9) -- (BS1);
        \draw [color = lightgray, thick](U6) -- (BS1);
        \draw [color =  red, dotted, thick](U8) -- (BS1);
        \draw [color = lightgray, thick](U12) -- (BS1);
        \draw [color =  red, dotted, thick](U15) -- (BS1);
        \draw [color =  red, dotted, thick](U10) -- (BS1);
        \draw [color =  red, dotted, thick](U7) -- (BS1);
        
        \draw [color =  red, dotted, thick](U3) -- (BS3);
        \draw [color = lightgray, thick](U11) -- (BS3);

        \node[right] at (0.4,3.1) {\footnotesize $R_2 \geq 300$};
        \draw[-Latex]  plot[smooth, tension=.7] coordinates {(0.6,2.1) (1,2.3) (1.1,2.8)};

        \end{tikzpicture}%
        }}
    \caption{Different failure models, where the red dotted, thick links have failed. Green users are still connected, yellow users are partially connected and red users are disconnected from the network. }
    \label{tikz:failure_models}
\end{figure*}

%% file: 5.discussion.tex
Our analysis relies on two key assumptions. First, BSs and users are homogeneously distributed according to a Poisson point process. Second, the BSs are allocated separate channels hence the interference from the adjacent cells is zero. 
While the first assumption is necessary for analytical tractability and it is a typical assumption in the literature, a more realistic model considering heterogeneous PPP can represent the current network deployments better. Regarding the second assumption, this assumption holds under appropriate frequency planning, e.g., with larger cell reuse cluster sizes. An interesting direction is to explore whether the observed trends are also valid in the existence of co-channel interference.



We also assumed that the degree of a BS is independent of the distance between a user and a BS. However, as Fig.~\ref{fig:capacity_nofailure} shows, this may not hold. For further research, one could look into the dependency between these two quantities to find better analytical expressions. However, it can also be seen that for larger values of $k$ this error becomes significantly smaller which means that for MC, this assumption is reasonable. 

We believe that there may be other policies that can make multi-connectivity beneficial to the user, also in terms of an increased channel capacity. An example of this could be a heterogeneous association scheme for MC in which other properties of a user are also considered in MC decision. For instance, cell-edge users can have more connections to prevent outage for this users and to ensure higher throughput.  Similarly, more sophisticated BS association schemes in the literature, e.g., BS load-aware, can be considered rather than the simple policy of connecting to the $k$ closest BSs. Designing policies that ensure that MC is exploited in such a way that all users benefit from it is therefore an important topic for further research. Finally, other interesting directions are the analysis of the overhead of MC in terms of scheduling complexity and message exchange among the BSs serving a user as well as other MC policies such as load balancing MC or packet duplication MC.

%% file: 6.conclusion.tex
In this paper, we have investigated the performance of multi-connected networks in terms of their per-user throughput and outage probabilities, using a model based on stochastic geometry. Interestingly, we observed that the per-user throughput of the network decreases for higher degrees of MC, i.e., when every user connects to multiple BSs. This is in contrast with previous results for the spectral efficiency~\cite{Gapeyenko2019OnDeployments}, which increased for moderate levels of MC. This shows that it is important to take the multi-user perspective into account when designing more efficient methods of obtaining multi-connectivity. While we showed that MC decreases the average channel capacity, we also show that some users still obtain a higher channel capacity under MC compared to single connectivity. 
Future research directions include investigation of different association schemes as well as different MC implementations such as as packet-splitting or load-balancing.

%% file: 7.appendix.tex
\begin{proof}
We want to find the expectation of the logarithm of $SNR_j$ for $j \in [k]$, which is the following integral:
\begin{IEEEeqnarray}{rCl}
    \E(\log_2(1+SNR_j)) &=& \frac{1}{\ln(2)}\int_0^\infty \ln(1+x) f_{SNR_j}(x) \dx \nonumber \\
    &=& \frac{1}{\ln(2)}\left( I_1 + I_2 + I_3\right) + R, \label{eq:expectation_log_snr}
\end{IEEEeqnarray}
where $R = \log_2(1+c)\left(1 - \frac{\Gamma(k, \lbs\pi)}{\Gamma(k + 1)} \right)$, the point mass in $x = c$, $f_{SNR_j}(x)$ is given in \eqref{eq:snr_pdf} and $I_1$, $I_2$ and $I_3$ are defined as:
\begin{IEEEeqnarray}{rCl}
    I_1 &=& \int_0^1 \ln(1+x) f_{SNR_j}(x)\dx, \label{eq:I1}\\
    I_2 &=& \int_1^c \ln\left(1+\frac{1}{x}\right)f_{SNR_j}(x)\dx, \label{eq:I2}\\
    I_3 &=& \int_1^c \ln(x)f_{SNR_j}(x)\dx \label{eq:I3}
\end{IEEEeqnarray}

To simplify notation, we define $\phi = \lbs \pi c^{\frac{2}{\alpha}}$. To find an expression for $I_1$ and $I_2$, we use a Taylor expansion for the logarithm:
\begin{IEEEeqnarray}{rCll}
    \ln(1+x) = \begin{cases}
        \sum_{i=0}^\infty \frac{(-1)^i}{i + 1} x^{i+1},  &0 \leq x \leq 1\\
        \ln(x) + \sum_{i=0}^\infty \frac{(-1)^i }{i+1}x^{-(i+1)}, & x > 1
    \end{cases}\nonumber
\end{IEEEeqnarray}
Therefore,
\begin{IEEEeqnarray}{rCl}
    I_1 &=& \frac{2\phi^j}{\alpha \Gamma(j)} \int_0^1 \ln(1+x) x^{-\frac{2j}{\alpha}-1} e^{-\phi x^{-\frac{2}{\alpha}}} \dx \nonumber \\
    &=& \frac{2\phi^j}{\alpha \Gamma(j)} \sum_{i=0}^\infty \frac{(-1)^i}{i + 1} \int_0^1 x^{-\frac{2j}{\alpha} + i} e^{-\phi x^{-\frac{2}{\alpha}}} \dx.
\end{IEEEeqnarray}
We use the change of variables $y = \phi x^{-\frac{2}{\alpha}}$:
\begin{IEEEeqnarray}{rCl}
    I_1 &=& \frac{1}{\Gamma(j)} \sum_{i=0}^\infty \frac{(-1)^i\phi^{\frac{\alpha}{2}(i + 1)}}{i + 1} \int_{\phi}^{\infty} y^{-\frac{\alpha}{2}(i + 1) + j - 1} e^{-y} \dy \nonumber \\
    &=&  \frac{1}{\Gamma(j)} \sum_{i=0}^\infty \frac{(-1)^i\phi^{\frac{\alpha}{2}(i + 1)}}{i + 1}  \Gamma\left(-\frac{\alpha}{2}(i+1) + j, \phi\right).\IEEEeqnarraynumspace \label{eq:i1}
\end{IEEEeqnarray}

For $I_2$, we use the same approach:
\begin{IEEEeqnarray}{rCl}
    I_2 &=& \frac{2\phi^j}{\alpha \Gamma(j)} \int_1^c \ln\left(1+\frac{1}{x}\right) x^{-\frac{2}{\alpha}-1} e^{-\phi x^{-\frac{2}{\alpha}}} \dx \nonumber \\
    &=& \frac{2\phi^j}{\alpha \Gamma(j)} \sum_{i=0}^\infty \frac{(-1)^i}{i + 1} \int_1^c x^{-\frac{2}{\alpha} - i - 2} e^{-\phi x^{-\frac{2}{\alpha}}} \dx\nonumber \\
    &=& \frac{1}{\Gamma(j)} \sum_{i=0}^\infty \frac{(-1)^i\phi^{-\frac{\alpha}{2}(i+1)}}{i + 1} \int_{\lbs \pi  }^{\phi} y^{\frac{\alpha}{2}(i+1) + j -1} e^{-y} \dy \nonumber \\
    &=&  \frac{1}{\Gamma(j)} \sum_{i=0}^\infty \frac{(-1)^i\phi^{-\frac{\alpha}{2}(i+1)}}{i + 1} \left(\Gamma\left(\frac{\alpha}{2}(1+i) +j, \lbs \pi \right) {-}  \Gamma\left(\frac{\alpha}{2}(1+i) + j, \phi\right)\right),\IEEEeqnarraynumspace \label{eq:i2}
\end{IEEEeqnarray}
with the same change of variables as we used for $I_1$.

For $I_3$ we use a slightly different approach:
\begin{IEEEeqnarray}{rCl}
    I_3 &=& \frac{2\phi^j}{\alpha \Gamma(j)} \int_1^c \ln(x) x^{-\frac{2j}{\alpha}-1} e^{-\phi x^{-\frac{2}{\alpha}}} \dx \nonumber \\
    &=& \frac{\alpha}{2\Gamma(j)} \int_{\lbs \pi}^{ \phi} \ln\left(\frac{\phi}{y}\right) y^{j-1} e^{-y} \dy \nonumber \\
    &=& \frac{\alpha}{2\Gamma(j)} \int_{\lbs \pi}^{ \phi} \left(\ln(\phi) - \ln(y)\right)y^{j-1} e^{-y} \dy \nonumber \\
    &=& \frac{\alpha}{2\Gamma(j)}\Bigg(\ln(\phi)\left(\Gamma(j, \lbs \pi) - \Gamma(j, \phi)\right) - \int_{\lbs \pi}^{ \phi}\ln(y) y^{j-1}e^{-y} \dy \Bigg) \nonumber\\
    &=& \frac{\alpha}{2\Gamma(j)}\Bigg(\ln(\phi)\left(\Gamma(j, \lbs \pi)- \Gamma(j, \phi)\right) - \frac{\text{d}}{\text{d}a} \left[ \Gamma(a, \phi) - \Gamma(a, \lbs \pi)\right]_{a = j} \Bigg), \label{eq:i3}
\end{IEEEeqnarray}
with the same changes of variables as in $I_1$ and $I_2$ and we can derive the last equation since this integrand is the derivative of the incomplete gamma function.

In \eqref{eq:elogsnr}, we filled in \eqref{eq:expectation_log_snr} with Equations \eqref{eq:i1}, \eqref{eq:i2} and \eqref{eq:i3}, with $G(j, \phi)$ as defined in \eqref{eq:defineG}, which gives the result we wanted to prove.
\begin{IEEEeqnarray}{rCl}
    \E(\log_2(1+SNR_j)) &=& \frac{1}{\ln(2)\Gamma(j)} G(j, \phi) +\log_2(1+c) \left(1 - \frac{\Gamma(j, \lbs \pi)}{\Gamma(j + 1)}\right) \nonumber \\
    &&+ \frac{\alpha}{2 \ln(2) \Gamma(j )}\bigg(\ln(\phi)\left(\Gamma(j, \lbs \pi) - \Gamma(j, \phi)\right) - \frac{\text{d}}{\text{d}a} \left[ \Gamma(a, \lbs \pi) - \Gamma(a, \phi)\right]_{a = j} \bigg).\label{eq:elogsnr}\IEEEeqnarraynumspace
\end{IEEEeqnarray}
    \begin{IEEEeqnarray}{rCll}
    G(j, \phi) &=& \sum_{i=0}^\infty \frac{(-1)^i\phi^{\frac{\alpha}{2}(i+1)}}{i + 1} \Gamma\left(\frac{-\alpha}{2}(i{+}1) {+} j, \phi\right) \nonumber \\
    &&{+} \frac{(-1)^i\phi^{-\frac{\alpha}{2}(i+1)}}{i + 1}  \bigg(\Gamma\left(\frac{\alpha}{2}(1{+}i) {+} j, \lbs \pi \right) {-}  \Gamma\left(\frac{\alpha}{2}(1{+}i) {+} j, \phi\right)\bigg).\IEEEeqnarraynumspace\label{eq:defineG}
\end{IEEEeqnarray}
\end{proof}

\subsection{Approximation}
In high-SNR regime, we can approximate \eqref{eq:expectation_log_snr} by omitting $I_1$ and $I_2$, as $f_{SNR}$ will be small between 0 and 1 and $\ln(1+SNR^{-1})$ will be as well. Figure \ref{fig:I1I2I3R} shows that indeed $I_1$ and the rest $R$, as defined in \eqref{eq:I1} - \eqref{eq:I3} are dominant with respect to $I_1$ and $I_2$ for large values of $\lbs$, the approximation $I_3 + R$ would work well.

\begin{figure}[h!]
\begin{subfigure}{0.5\textwidth}
    \includegraphics[width = 0.9\textwidth]{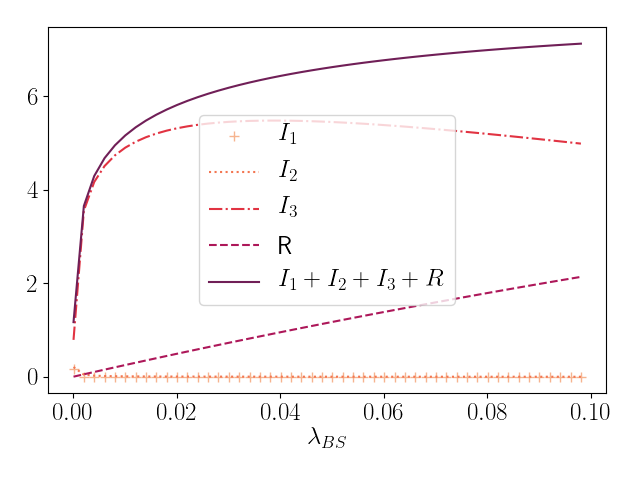}
    \caption{$j=1$}
    \label{sfig:sumIntegral_alpha2_k1}
\end{subfigure}
\hfill
\begin{subfigure}{0.5\textwidth}
    \includegraphics[width=0.9\textwidth]{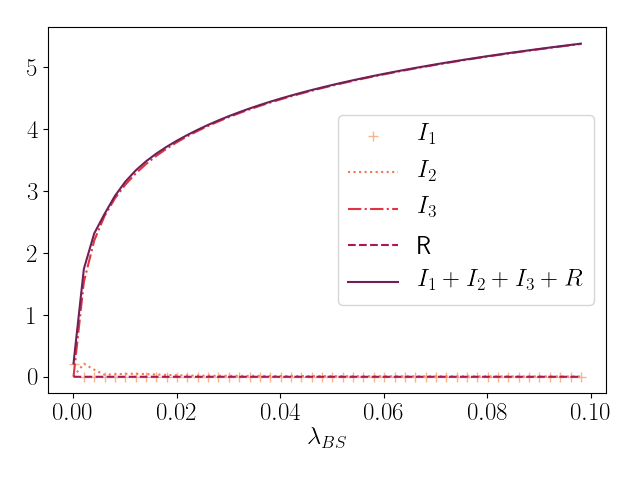}
    \caption{$j=5$}
    \label{sfig:sumIntegral_alpha2_k5}
\end{subfigure}
\caption{$\E(\log_2(1+SNR_j)$ for different values of $j$, as a sum of $I_1, I_2, I_3$ and $R$ with $c = 1.28 \cdot 10^4$ and $\alpha = 2$.}
\label{fig:I1I2I3R}
\end{figure}
We start with the approximation of $I_3$, as defined in \eqref{eq:i3}. We can express the incomplete gamma function $\Gamma(s, x)$ for $s \in \mathbb{N}$ as~\cite{Amore2005AsymptoticFunction}:
\begin{IEEEeqnarray}{rCl}
    \Gamma(s,x) = \Gamma(s)e^{-x}\sum_{i=0}^{s-1}\frac{x^i}{i!}, \label{eq:approximation_gamma}
\end{IEEEeqnarray}
which holds for $|x| < 1$, while
for large values of $x$, the incomplete gamma function tends to zero.

For the second term of $I_3$ in~\eqref{eq:i3}, we need to find the approximation of the derivative of the gamma function, which is given in \cite{geddes1990evaluation}:
\begin{IEEEeqnarray}{rCl}
    \frac{\text{d}}{\ds} \Gamma(s, x) &=& \ln(x)\Gamma(s,x) - x \frac{\text{d}}{\dt}\Gamma(s-t)x^{t-1}\Big|_{t=0} + \sum_{i=0}^\infty \frac{(-1)^i x^{s+i}}{i!(s+i)^2}\nonumber \\
    &=& \ln(x) \Gamma(s,x) - \Gamma(s)(\ln(x) + \gamma - H_{s-1}) + \sum_{i=0}^\infty \frac{(-1)^i x^{s+i}}{i!(s+i)^2}, \label{eq:approximate_derivative_gamma}
\end{IEEEeqnarray}
where $H_{s} = \sum_{i=1}^s \frac{1}{i}$ is the harmonic number.\\

We can now fill in \eqref{eq:i3} using \eqref{eq:approximation_gamma} and \eqref{eq:approximate_derivative_gamma}:
\begin{IEEEeqnarray}{rCl}
    I_3 &=& \frac{\alpha}{2}\Bigg( \big(\ln(\phi) -\ln(\lbs \pi)\big)\Bigg(e^{-\lbs \pi}\sum_{i=0}^{j-1}\frac{(\lbs \pi)^i}{i!} - 1\Bigg) {+} \sum_{i=0}^\infty \frac{(-1)^i (\lbs \pi)^{j+i}}{i!(j+i)^2\Gamma(j)} {-}  \sum_{i=0}^\infty \frac{(-1)^i \phi^{j+i}}{i!(j+i)^2\Gamma(j)} \Bigg) \nonumber \\
    &=&  \ln(c)\left(e^{-\lbs \pi}\sum_{i=0}^{j-1}\frac{(\lbs \pi)^i}{i!} - 1\right) + \frac{\alpha}{2 \Gamma(j)} \sum_{i=0}^\infty \frac{(-1)^i \left((\lbs \pi)^{j+i} - \phi^{j+i}\right)}{i!(j+i)^2} .\label{eq:I3exact}
\end{IEEEeqnarray}

Assuming $\lbs \pi < 1$, we can approximate \eqref{eq:I3exact} for $j > 1$:
\begin{IEEEeqnarray}{rCl}
    I_3 &=& 
        \ln(c)\left(e^{-\lbs \pi}\left(1+\lbs \pi + O\left((\lbs \pi)^2\right)\right) - 1\right)- \frac{\alpha}{2\Gamma(j)}\sum_{i=0}^\infty \frac{(-1)^i \phi^{j+i}}{i!(j+i)^2\Gamma(j)}\cdot \left(1 + O\left(c^{-\frac{2j}{\alpha}}\right)\right) \nonumber \\
        &=& 
        \ln(c)\left(e^{-\lbs \pi}\left(1{+}\lbs \pi {+} O\left((\lbs\pi)^2\right)\right) - 1\right) {+}\frac{\alpha \phi^j}{2j^2\Gamma(j)}\, _2F_2(\{j, j\}, \{1{+}j, 1{+}j\}, {-}\phi) \Big(1 {+} O\Big(c^{-\frac{2j}{\alpha}}\Big)\Big),\ruimte\label{eq:I3_approximation_kgreater1}
\end{IEEEeqnarray}
where $_2F_2(a, b, x)$ is the hypergeometric function that we can approximate using Theorem 3.1 of \cite{volkmer2014note}:
\begin{IEEEeqnarray}{rCl}
    _2F_2(\{j, j\}, \{1+j, 1+j\}, -\phi) &=& \frac{\left(\Gamma(j+1)\right)^2}{\left(\Gamma(j)\right)^2}\int_0^1 e^{-\phi y}(1-y)y^{j-1}\,_2F_1\left(\{1,1\}, \{2\}, 1{-}y\right) \dy \nonumber \\
    &=& \frac{j^2}{\phi^j}\left( \Gamma(j)\left(\ln(\phi) + \gamma - H_{j-1}\right) +  G_{2,3}^{3,0}\left(\phi\left|
    \begin{array}{c}
     1,1 \\
     0,0,j \\
    \end{array}
    \right.\right)\right) \nonumber \\
    &=& \frac{j^2 \Gamma(j)}{\phi^j}  \left(\ln(\phi) + \gamma - H_{j-1}\right) + O\left( j^2 \phi^{-2j}\right),\ruimte \label{eq:Hypergeometric_approximation}
\end{IEEEeqnarray}
where $G_{2,3}^{3,0}\left(\phi\left|
    \begin{array}{c}
     1,1 \\
     0,0,j \\
    \end{array}\right.\right)$ is the Meijer-G function, which is always smaller than $\phi^{-j}$.

The approximation in \eqref{eq:Hypergeometric_approximation} can be filled in in \eqref{eq:I3_approximation_kgreater1} to obtain a final approximation of $I_3$ for $j > 1$:
\begin{IEEEeqnarray}{rCl}
    I_3 &=& \ln(c)\left(e^{-\lbs \pi}\left(1+\lbs \pi + O\left((\lbs\pi)^2\right)\right) - 1\right) \nonumber \\
    &&+ \frac{\alpha}{2} (\ln(\phi)  {+} \gamma {-} H_{j-1})  \Big(1 {+} O\Big(c^{-\frac{2j}{\alpha}}\Big)\Big) {+} O\Big(\Gamma(j)^{-1} c^{-\frac{2j}{\alpha}}\phi^{-j}\Big) \ruimte\label{eq:I_3_approximation_kgeq1}
\end{IEEEeqnarray}

For $j = 1$, $I_3$ can directly be calculated from \eqref{eq:I3}:
\begin{IEEEeqnarray}{rCl}
    I_3 &=& \ln(c) e^{-\lbs \pi}  + \frac{\alpha}{2}\left(\Gamma(0, \phi) - \Gamma(0, \lbs \pi)\right) \nonumber \\
     &=& \ln(c) e^{-\lbs \pi} + \frac{\alpha}{2}\left(-\frac{1}{\phi} e^{-\phi} + \ln(\lbs \pi) + \gamma -\lbs \pi\right) + O\left(\phi^{-1}\right) + O\left((\lbs \pi)^2\right),\label{eq:I3_approximation_k1}
\end{IEEEeqnarray}
using the series expansion around $0$ and around $\infty$ for $\Gamma(0, x)$.\\

As $\E(\log_2(1+SNR_j)) = \frac{1}{\ln(2)} (I_1 + I_2 + I_3) + R$, we still need to approximate $I_1, I_2$ and $R$. To approximate $R$, we again approximate $\Gamma(k, x)$ by the first 2 terms of \eqref{eq:approximation_gamma}:
\begin{IEEEeqnarray}{rCl}
    R = \begin{cases}
        \log_2(c)\big(1 {-} e^{-\lbs \pi}\big) + O\left(c^{-1}\log_2(c)\right), \quad j = 1,\\
        \log_2(c)\big(1 {-} e^{-\lbs \pi}(1{+}\lbs \pi {+} O\big((\lbs \pi)^2\big) {+} O\big(c^{-1}\big)\big),\label{eq:approximation_R}
    \end{cases}
\end{IEEEeqnarray}
as the second equation holds for $j > 1$. For $I_1$ and $I_2$, we find an upper bound as Figure \ref{fig:I1I2I3R} shows these integrals are negligible in high-SNR regime:
\begin{IEEEeqnarray}{rCl}
    I_1 &\leq& \ln(2) \cdot \P(SNR_j \leq 1) = \ln(2) \frac{\Gamma(j, \phi)}{\Gamma(j)} = O\left(e^{-\phi}\phi^j\right), \IEEEeqnarraynumspace\label{eq:I1_approx}\\
    I_2 &\leq& \ln\left(1+\frac{1}{\E(SNR_j)}\right) = O\left(\E(SNR_j^{-1})\right).\label{eq:I2_approx}
\end{IEEEeqnarray}

Now, we approximate $\E(\log_2(1+SNR_j))$ by using \eqref{eq:I3_approximation_k1} and \eqref{eq:I_3_approximation_kgeq1} together with \eqref{eq:approximation_R} - \eqref{eq:I2_approx}:
\begin{IEEEeqnarray}{rCl}
    \E(\log_2(1+SNR_1)) &=& \frac{\alpha}{2 \ln(2)}\left(\ln(\phi) + \gamma - \lbs \pi \right)  + \delta_1,\IEEEeqnarraynumspace\label{eq:approx_snr_k1}\\
    \E(\log_2(1+SNR_j)) &=& \frac{\alpha }{2 \ln(2)} \left(\ln(\phi) + \gamma - H_{j-1}\right) + \delta_j, \IEEEeqnarraynumspace\label{eq:approx_snr}
\end{IEEEeqnarray}
for $j > 1$ and where $\gamma$ is Euler's constant and for the following variables $\delta_1$ and $\delta_j$:
\begin{IEEEeqnarray}{rCl}
    \delta_1 &=&  O\left( \phi^{-1}\right) {+} O\left( (\lbs \pi)^2\right) {+} O\left(\ln(c) c^{-1}\right) {+} O\left(\E(SNR_j^{-1})\right)\nonumber\\
    \delta_j &=& O\left(c^{-\frac{2j}{\alpha}} \ln(\phi)\right) {+} O\left(\Gamma(j)^{-1} c^{-\frac{2j}{\alpha}}\phi^{-j}\right) {+} O\left(\ln(c) c^{-1}\right)+ O\left(\E(SNR_j^{-1})\right),\ruimte \nonumber
\end{IEEEeqnarray}
where the latter holds for $j > 1$.

%% file: main-single_column.bbl
\begin{thebibliography}{10}
\bibitem{Simsek20165G-EnabledInternet}
M.~Simsek, A.~Aijaz, M.~Dohler, J.~Sachs, and G.~Fettweis, ``{5G-Enabled
  Tactile Internet},'' \emph{IEEE Journal on Selected Areas in Communications},
  vol.~34, no.~3, pp. 460--473, 3 2016.

\bibitem{Wolf2019HowMulti-Connectivity}
A.~Wolf, P.~Schulz, M.~Dorpinghaus, J.~C.~S. Santos~Filho, and G.~Fettweis,
  ``{How Reliable and Capable is Multi-Connectivity?}'' \emph{IEEE Transactions
  on Communications}, vol.~67, no.~2, pp. 1506--1520, 2 2019.

\bibitem{Suer2020Multi-ConnectivityOverview}
M.~T. Suer, C.~Thein, H.~Tchouankem, and L.~Wolf, ``{Multi-Connectivity as an
  Enabler for Reliable Low Latency Communications - An Overview},'' \emph{IEEE
  Communications Surveys and Tutorials}, vol.~22, no.~1, pp. 156--169, 1 2020.

\bibitem{petrov2017dynamic}
V.~Petrov, D.~Solomitckii, A.~Samuylov, M.~A. Lema, M.~Gapeyenko,
  D.~Moltchanov, S.~Andreev, V.~Naumov, K.~Samouylov, M.~Dohler \emph{et~al.},
  ``Dynamic multi-connectivity performance in ultra-dense urban mmwave
  deployments,'' \emph{IEEE Journal on Selected Areas in Communications},
  vol.~35, no.~9, pp. 2038--2055, 2017.

\bibitem{Gapeyenko2019OnDeployments}
M.~Gapeyenko, V.~Petrov, D.~Moltchanov, M.~R. Akdeniz, S.~Andreev, N.~Himayat,
  and Y.~Koucheryavy, ``{On the degree of multi-connectivity in 5G
  millimeter-wave cellular urban deployments},'' \emph{IEEE Transactions on
  Vehicular Technology}, vol.~68, no.~2, pp. 1973--1978, 2 2019.

\bibitem{ghatak2020elastic}
G.~Ghatak, Y.~Sharma, K.~Zaid, and A.~U. Rahman, ``{Elastic multi-connectivity
  in 5G networks},'' \emph{Physical Communication}, vol.~43, p. 101176, 2020.

\bibitem{sharma2020statistical}
Y.~Sharma and G.~Ghatak, ``{A Statistical Characterization of SINR Coverage and
  Network Throughput with Macro-Diversity},'' in \emph{2020 IEEE 21st
  International Symposium on" A World of Wireless, Mobile and Multimedia
  Networks"(WoWMoM)}.\hskip 1em plus 0.5em minus 0.4em\relax IEEE, 2020, pp.
  197--204.

\bibitem{wolf2018rate}
A.~Wolf, P.~Schulz, D.~{\"O}hmann, M.~D{\"o}rpinghaus, and G.~Fettweis,
  ``{Rate-reliability tradeoff for multi-connectivity},'' in \emph{2018 IEEE
  Wireless Communications and Networking Conference (WCNC)}.\hskip 1em plus
  0.5em minus 0.4em\relax IEEE, 2018, pp. 1--6.

\bibitem{perdomo2020user}
J.~Perdomo, M.~Ericsson, M.~Nordberg, and K.~Andersson, ``{User Performance in
  a 5G Multi-connectivity Ultra-Dense Network City Scenario},'' in \emph{2020
  IEEE 45th Conference on Local Computer Networks (LCN)}.\hskip 1em plus 0.5em
  minus 0.4em\relax IEEE, 2020, pp. 195--203.

\bibitem{gerasimenko2019capacity}
M.~Gerasimenko, D.~Moltchanov, M.~Gapeyenko, S.~Andreev, and Y.~Koucheryavy,
  ``{Capacity of multiconnectivity mmWave systems with dynamic blockage and
  directional antennas},'' \emph{IEEE Transactions on Vehicular Technology},
  vol.~68, no.~4, pp. 3534--3549, 2019.

\bibitem{pirmagomedov2019performance}
R.~Pirmagomedov, D.~Moltchanov, V.~Ustinov, M.~N. Saqib, and S.~Andreev,
  ``{Performance of mmwave-based mesh networks in indoor environments with
  dynamic blockage},'' in \emph{International Conference on Wired/Wireless
  Internet Communication}.\hskip 1em plus 0.5em minus 0.4em\relax Springer,
  2019, pp. 129--140.

\bibitem{suer2020evaluation}
M.-T. Suer, C.~Thein, H.~Tchouankem, and L.~Wolf, ``{Evaluation of
  Multi-Connectivity schemes for URLLC traffic over WiFi and LTE},'' in
  \emph{2020 IEEE Wireless Communications and Networking Conference
  (WCNC)}.\hskip 1em plus 0.5em minus 0.4em\relax IEEE, 2020, pp. 1--7.

\bibitem{suer2020impact}
------, ``{Impact of Link Heterogeneity and Link Correlation on
  Multi-Connectivity Scheduling Schemes for Reliable Low-Latency
  Communication},'' in \emph{2020 IEEE International Conference on
  Communications Workshops (ICC Workshops)}.\hskip 1em plus 0.5em minus
  0.4em\relax IEEE, 2020, pp. 1--6.

\bibitem{suer2020reliability}
------, ``{Reliability and Latency Performance of Multi-Connectivity Scheduling
  Schemes in Multi-User Scenarios},'' in \emph{2020 32nd International
  Teletraffic Congress (ITC 32)}.\hskip 1em plus 0.5em minus 0.4em\relax IEEE,
  2020, pp. 73--80.

\bibitem{ElSawy2017ModelingTutorial}
H.~ElSawy, A.~Sultan-Salem, M.~S. Alouini, and M.~Z. Win, ``{Modeling and
  Analysis of Cellular Networks Using Stochastic Geometry: A Tutorial},''
  \emph{IEEE Communications Surveys and Tutorials}, vol.~19, no.~1, pp.
  167--203, 1 2017.

\bibitem{shannon2001mathematical}
C.~E. Shannon, ``A mathematical theory of communication,'' \emph{ACM SIGMOBILE
  mobile computing and communications review}, vol.~5, no.~1, pp. 3--55, 2001.

\bibitem{stegehuis2021degree}
C.~Stegehuis and L.~Weedage, ``Degree distributions in {AB} random geometric
  graphs,'' 2021.

\bibitem{jain1984quantitative}
R.~K. Jain, D.-M.~W. Chiu, W.~R. Hawe \emph{et~al.}, ``A quantitative measure
  of fairness and discrimination,'' \emph{Eastern Research Laboratory, Digital
  Equipment Corporation, Hudson, MA}, 1984.

\bibitem{Amore2005AsymptoticFunction}
P.~Amore, ``{Asymptotic and exact series representations for the incomplete
  Gamma function},'' \emph{EPL (Europhysics Letters)}, Vol.~71, no.~1, pp. 1--7, 2005.

\bibitem{geddes1990evaluation}
K.~O. Geddes, M.~L. Glasser, R.~A. Moore, and T.~C. Scott, ``Evaluation of
  classes of definite integrals involving elementary functions via
  differentiation of special functions,'' \emph{Applicable Algebra in
  Engineering, Communication and Computing}, vol.~1, no.~2, pp. 149--165, 1990.

\bibitem{volkmer2014note}
H.~Volkmer and J.~J. Wood, ``A note on the asymptotic expansion of generalized
  hypergeometric functions,'' \emph{Analysis and Applications}, vol.~12,
  no.~01, pp. 107--115, 2014.

\end{thebibliography}
